\documentclass[american,english,journal]{IEEEtran}
\usepackage[T1]{fontenc}
\usepackage[latin9]{inputenc}
\usepackage{amsmath}
\usepackage{amsthm}
\usepackage{amssymb}
\usepackage{graphicx}

\makeatletter
\usepackage{babel}
\usepackage{amsfonts}
\usepackage{cite}
\usepackage{array}
\usepackage{algorithm}
\usepackage{algorithmic}
\usepackage{subfigure}
\usepackage{stfloats}
\usepackage{graphicx}
\usepackage{float}
\usepackage[font=small,labelfont=normalfont]{caption}
\usepackage{url}
\usepackage{siunitx}

\makeatother

\makeatother

\theoremstyle{plain}
\newtheorem{lem}{\protect\lemmaname}
\newtheorem{thm}{\protect\theoremname}
\usepackage{babel}
\addto\captionsamerican{\renewcommand{\lemmaname}{Lemma}}
\addto\captionsamerican{\renewcommand{\theoremname}{Theorem}}
\addto\captionsenglish{\renewcommand{\lemmaname}{Lemma}}
\addto\captionsenglish{\renewcommand{\theoremname}{Theorem}}
\providecommand{\lemmaname}{Lemma}
\providecommand{\theoremname}{Theorem}

\begin{document}
\title{Blockchain-based Proportional Fair Scheduling for Multi-Operator O-RAN}
\author{Kun~Huang, Xintong~Ling,~\IEEEmembership{Member,~IEEE}, Meining~Wu,
Jiaheng~Wang,~\IEEEmembership{Senior Member,~IEEE},\\
Zhi~Ding,~\IEEEmembership{Fellow,~IEEE}, and Xiqi~Gao,~\IEEEmembership{Fellow,~IEEE}\vspace{-0.6cm}
\thanks{K. Huang, X. Ling, M. Wu, J. Wang, X. Gao are with the National Mobile
Communications Research Laboratory, Southeast University, Nanjing
210096, China (e-mail: kunhuang@seu.edu.cn, meining.wu@outlook.com,
xtling@seu.edu.cn, jhwang@seu.edu.cn, and xqgao@seu.edu.cn). K. Huang,
X. Ling, J. Wang, and X. Gao are also with the Purple Mountain Laboratories,
Nanjing 210023, China. J. Wang is also with the Cyber Science and
Engineering, Southeast University.}\thanks{Z. Ding is with the Department of Electrical and Computer Engineering,
University of California at Davis, Davis, CA 95616 USA (e-mail: zding@ucdavis.edu).}}
\maketitle
\begin{abstract}
The openness and disaggregation of Open radio access network (O-RAN)
facilitate resource sharing and coordination across networks, creating
new demands for efficient and trustworthy cross-operator scheduling.
However, such scheduling is beyond the scope and capability of conventional
proportional fair scheduling (PFS), which lacks mechanisms for establishing
trust among independent operators. To fulfill this gap, we propose
the blockchain-based proportional fair scheduling (BC-PFS) that enables
trustworthy inter-network coordination and resource pooling across
operators in O-RAN. Specifically, we design four core smart contracts
including registration, status reporting, scheduling, and settlement
contracts with corresponding Solidity implementations to ensure trustworthy
on-chain execution. Theoretically, to evaluate the BC-PFS performance,
we develop an analytical framework to derive the user average throughput
via both probabilistic and ordinary differential equation (ODE) approaches,
and provide a simplified closed-form solution. Based on the above
performance assessment, we quantify the pooling effect in O-RAN achieved
through trustworthy cross-operator collaboration via BC-PFS, and point
out that this effect grows monotonically in both the numbers of operator
networks and users. Simulations validate the theoretical analysis
and show the performance of the BC-PFS in O-RAN. 
\end{abstract}


\section{Introduction}

As the telecommunication industry transitions from 5G to the upcoming
6G, the growing diversity of applications demands seamless communication
with guaranteed quality of service (QoS) across heterogeneous networks
\cite{Fu2024}. Open radio access network (O-RAN) has emerged as a
next-generation network architecture to meet these requirements through
openness, disaggregation, intelligence, and programmability \cite{whitePaper2018,Polese2023}.
By leveraging network function virtualization and open interfaces,
O-RAN breaks the closed nature of traditional architectures and enables
broader participation from multiple operators, vendors, and third-party
applications. The growing demand for spectrum and RAN infrastructure
sharing in O-RAN necessitates efficient and trustworthy coordination
across different networks and operators \cite{Oransharing2023,Damnjanovic2024,Javed2025}. 

However, the trust foundation for this collaboration is often limited
in O-RAN due to the disaggregated network components, multi-vendor
implementations, and separate network management \cite{Giupponi2022}.
Specifically, cross-operator collaborative scheduling requires multiple
stakeholders to trustfully exchange user and resource information
and jointly execute their scheduling decisions. The reliability of
coordinated decisions and their execution is also challenging for
different parties to verify. Moreover, geographically fragmented and
heterogeneous networks further increase the difficulty of tracing
and auditing the scheduling process. Therefore, the lack of trust
mechanisms constrains reliable coordination across networks, ultimately
degrading overall network performance. These challenges underscore
the necessity for a scheduling paradigm that ensures both distributed
resource orchestration and trustworthy inter-network cooperation within
O-RAN. 

As the default downlink scheduling for 4G, proportional fair scheduling
(PFS), can balance system throughput and user fairness \cite{Astely2009}
and has demonstrated robust performance under diverse network conditions
\cite{Astely2009,Haque2023}. To extend PFS beyond isolated networks,
early multi-cell PFS approaches proposed in \cite{Zhou2011} coordinated
collaborative decisions across base stations with shared user information.
Subsequent studies further developed joint PFS strategies that maximize
the product of user priorities and extended them to multi-user selection
scenarios \cite{Gu2016,Li2018,Zhang2022a}. However, these schemes
inherently rely on information sharing among base stations or networks,
and generally assume trustworthy execution among participating entities
and truthful information exchange. Yet this assumption is hard to
sustain in multi-operator O-RAN, where participants lack an inherent
trust foundation for collaborative scheduling. As a result, the cross-network
coordination process is vulnerable to potential manipulation risks
and unreliable execution.

Remark that blockchain provides a promising trust infrastructure for
trustworthy inter-network collaboration through consensus-based validation
and smart contract execution \cite{Faisal2022}. By maintaining a
shared, tamper-resistant, and traceable ledger, blockchain offers
multiple operators a consistent and verifiable view of coordination
without relying on intermediaries \cite{Wu2026}. Consensus mechanisms
establish agreement on scheduling rules and states, while smart contracts
enable automatic and transparent execution of coordination strategies
\cite{Xu2024}. These capabilities have enabled blockchain to support
dynamic user scheduling and resource sharing across operators for
O-RAN \cite{Ling2019,Ling2025,Chen2025,Wang2022d}. 

Recent studies have explored blockchain to facilitate cross-operator
resource coordination and user association. Research in \cite{Ling2019,Ling2020a,Cao2023,Wang2022d,Ling2025a}
enables flexible cross-operator spectrum access through blockchain-based
decentralized coordination, while works like \cite{Femenias2024,Al2024}
utilized smart contracts for dynamic spectrum pricing and demand-driven
resource allocation. Beyond architectural designs, several studies
have sought to improve network performance through mathematical modeling
of user association \cite{Qian2025,Zhang2020d,Yu2026}. The study
in \cite{Zhang2022} introduced a proof-of-strategy consensus mechanism
that maximizes the global efficiency with optimal user association.
To ensure QoS during spectrum access, \cite{Cao2023} employed a repeated
game-theoretic framework to establish the trust relationship between
users and service providers. Similar optimization perspectives can
be found in \cite{Zhu2022,Ayepah2023}, where auctions and game theory
are employed to refine association strategies. Although these blockchain-based
schemes improve coordination transparency and auditability, they often
rely on complex iterative algorithms, which pose challenges for deployment
in large-scale and geographically distributed O-RAN. Therefore, it
calls for a low-complexity on-chain scheduling scheme. Meanwhile,
most modeling studies lack analytical frameworks to quantitatively
characterize the performance gains of cross-operator scheduling through
blockchain. The performance improvement under blockchain-enhanced
cooperation remains unclear.

In summary, despite these advances, several critical limitations remain.
On the one hand, existing cross-operator user scheduling algorithms
often lack explicit trust mechanisms for inter-network cooperation,
making reliable coordination difficult to sustain in multi-operator
O-RAN. On the other hand, most blockchain-based scheduling approaches
lack tight integration with the physical layer and often involve high
computational complexity. Furthermore, existing studies provide limited
analysis of cross-operator scheduling enabled by blockchain. The fundamental
questions remain unanswered, particularly regarding how to facilitate
efficient and trustworthy cross-operator scheduling and how many improvements
can be attained through blockchain-enhanced O-RAN.

To fulfill the above gaps, we propose the blockchain-based proportional
fair scheduling (BC-PFS) which integrates PFS to enable trustworthy
cross-operator user scheduling for O-RAN. BC-PFS preserves the low-complexity
and responsive adaptation of conventional PFS, and meanwhile it establishes
trusted coordination and verifiable execution among operators. Moreover,
we develop systematic analytical frameworks to assess the system performance
and quantify the pooling effect of O-RAN via BC-PFS. Ultimately, we
aim to address the fundamental questions about efficiency, fairness,
and measurable utility improvements in O-RAN via blockchain-enabled
scheduling across operators. The main contributions of this paper
are summarized as follows:
\begin{itemize}
\item We propose the BC-PFS that enables trustworthy user scheduling for
multi-operator O-RAN. Specifically, we design the BC-PFS with four
core smart contracts, namely registration, status reporting, scheduling,
and settlement contracts.
\item We establish a comprehensive analytical framework for the BC-PFS performance
evaluation and derive the user throughput through both probabilistic
modeling and ordinary differential equation (ODE) approaches. Furthermore,
we provide a simplified solution in a closed form by introducing a
relatively weak assumption. 
\item We quantitatively characterize the pooling effect introduced by BC-PFS
in O-RAN. We prove that the pooling effect is monotonically increasing
with respect to the numbers of operator networks and users per network,
demonstrating consistent performance improvement from the network
openness.
\item Through a series of simulations, we validate our theoretical derivations
and demonstrate the superiority of the BC-PFS compared to conventional
non-cooperative approaches. We provide the BC-PFS implementation based
on Solidity and JavaScript at \url{https://github.com/kunhuangseu/BC-PFS}.
\end{itemize}
The remainder of this paper is organized as follows. Section \ref{sec:System-model}
outlines the system model. Section \ref{sec:Revisit-PF-Scheduling}
revisits the PFS algorithm. Section \ref{sec:Blockchain-based-PF-Scheduling}
presents the proposed BC-PFS. Section \ref{sec:Performance-Analysis}
provides the mathematical performance analysis. Section \ref{sec:Blockchain-Pooling-Effect}
quantitatively evaluates the pooling effect. Section \ref{sec:Simulation-and-analysis}
presents the simulations. Finally, Section \ref{sec:Conclusion} concludes
the paper.

\section{System model\label{sec:System-model}}

\subsection{Multi-operator O-RAN\label{subsec:Inter-network-Scheduling-Model}}

\begin{figure}
\centering\includegraphics[width=0.33\paperwidth]{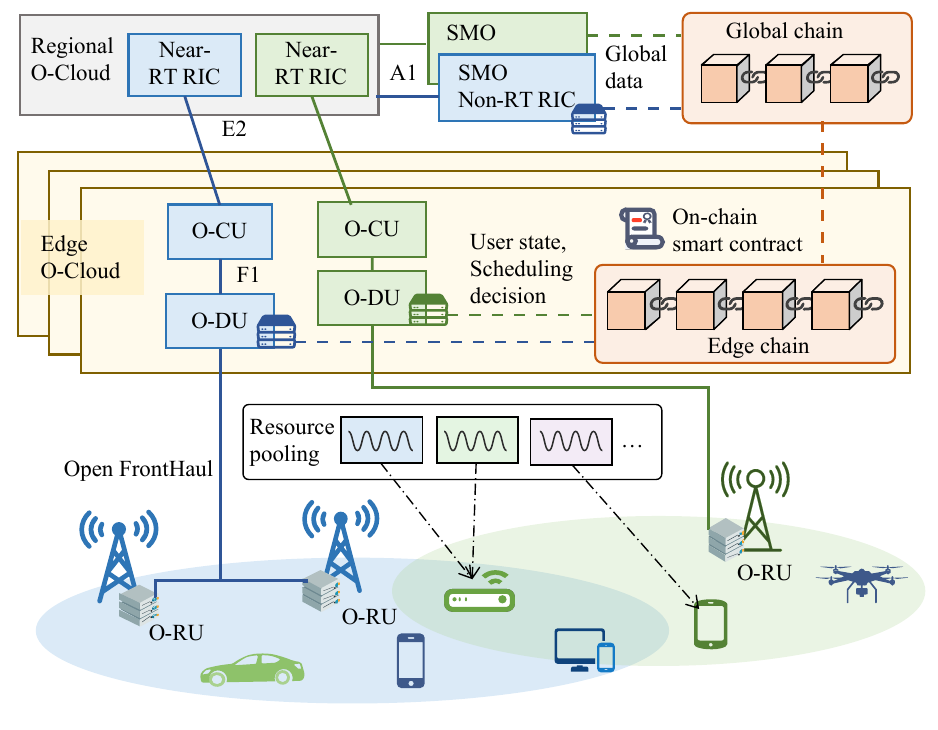}\caption{\foreignlanguage{american}{Illustration of the BC-PFS for multi-operator O-RAN.\label{fig:The-system-architecture}}}
\vspace{-2mm}
\end{figure}
We consider a $K$-operator O-RAN, each serving $N$ users. As illustrated
in Fig. \ref{fig:The-system-architecture}, O-RAN consists of functional
entities deployed across edge and regional O-RAN cloud (O-Cloud).
At edge O-Cloud, each operator maintains logically isolated O-RAN
distributed unit (O-DU) and O-RAN centralized unit (O-CU) instances.
The O-DU performs radio resource scheduling and connects to O-RAN
radio units (O-RUs), which provide radio access to user equipment
(UEs). At the regional O-Cloud, the near-real-time RAN intelligent
controller (Near-RT RIC) provides programmable control and coordination,
while the non-real-time RIC (Non-RT RIC) and service management and
orchestration (SMO) support long-term orchestration, policy management,
and network optimization.

The scheduling is based on a time-division allocation framework, where
time is partitioned into discrete, uniformly spaced slots of duration.
We assume that the channel coherence time is larger than a slot time.
Let $r_{nk}\left(t\right)$ denote the instantaneous rate provided
by operator $k$ to user $n$ at time slot $t$. Every operator $k$
has the corresponding licensed spectrum to provide services, and the
instantaneous rates $r_{nk_{1}}\left(t\right)$ and $r_{nk_{2}}\left(t\right)$
for any two distinct operators $k_{1}\neq k_{2}$ are mutually independent.
Furthermore, $\mu_{nk}\left(t\right)$ represents the throughput achieved
by user $n$ from operator $k$ up to time slot $t$. 

In the considered multi-operator O-RAN, a user can dynamically access
multiple operators for service, while each operator\textquoteright s
channel can serve only one user in each time slot. This association
is captured by a binary user-network association indicator $I_{nk}\left(t\right)\in\left\{ 0,1\right\} $,
where $I_{nk}\left(t\right)=1$ indicates user $n$ accesses the channel
of operator $k$ during slot $t$, while $I_{nk}\left(t\right)=0$
indicates the absence of scheduling for user $n$ by operator $k$
in that slot. Consequently, the achievable rate for user $n$ at slot
$t$ is given by $r_{n}\left(t\right)=\sum^{K}_{k=1}r_{nk}\left(t\right)I_{nk}\left(t\right)$.
The cumulative throughput for user $n$ aggregates contributions from
all operators and is given by $\mu_{n}\left(t\right)=\sum^{K}_{k=1}\mu_{nk}\left(t\right)$. 

\subsection{Wireless Channel Model\label{subsec:Wireless-Channel-Model}}

In this work, the wireless channel is modeled as the well-known Rayleigh
fading by assuming that the signal envelope follows a Rayleigh distribution
\cite{Rappaport2001}. It can effectively characterize rich scattering
environments lacking line-of-sight propagation paths. According to
the Rayleigh fading model, for user $n$, its instantaneous signal-to-noise
ratio (SNR) can be modeled as an exponentially distributed random
variable with the probability density function, $f_{\textrm{SNR}_{n}}\left(x\right)=\frac{1}{\gamma_{n}}\textrm{exp}\left(-\frac{x}{\gamma_{n}}\right)$,
where $\gamma_{n}$ represents the average SNR for user $n$. Based
on Shannon's formula, the instantaneous transmission rate $r_{n}$
can be expressed as:
\begin{equation}
r_{n}=B\log_{2}\left(1+\textrm{SNR}_{n}\right),\label{eq:SNR}
\end{equation}
with $B$ denoting the channel bandwidth. The statistical characteristics
of the rate are given by:
\begin{align}
\bar{r}_{n}= & \int^{\infty}_{0}B\log_{2}\left(1+\gamma_{n}x\right)\cdot\textrm{exp}\left(-x\right)dx,\label{eq:mean}\\
\sigma^{2}_{n}= & \int^{\infty}_{0}\left(B\log_{2}\left(1+\gamma_{n}x\right)\right)^{2}\cdot\textrm{exp}\left(-x\right)dx-\bar{r}^{2}_{n}.\label{eq:variance}
\end{align}
The average $\bar{r}_{n}$ and variance $\sigma^{2}_{n}$ characterize
the long-term average performance and random fluctuations of the user's
transmission rate, respectively. 

The logarithmic rate model in \eqref{eq:SNR} can be simplified to
a linear rate approximation through a first-order Taylor expansion
of Shannon\textquoteright s formula under low-SNR conditions. This
linear model has been widely used in the analysis of user scheduling
\cite{Kushner2004,Liu2011} and offers a practical trade-off between
accuracy and computational simplicity for low-SNR regimes. 

\section{Revisit Proportional Fair Scheduling\label{sec:Revisit-PF-Scheduling}}

In wireless networks, user scheduling fundamentally addresses a constrained
optimization problem aimed at selecting proper users to serve in limited
wireless resources to maximize network performance while adhering
to fairness and feasibility requirements \cite{Liew2008,Ma2019}.
This optimization framework typically involves defining criteria that
quantitatively characterize key system performance metrics, including
throughput and fairness. Among various fairness criteria, proportional
fairness has emerged as a pivotal concept due to its ability to achieve
a balance between network throughput and fairness.

We define the long-term average throughput of user $n$, denoted by
$\bar{\mu}_{n}=\underset{t\rightarrow\infty}{\text{lim}}E\left[\mu_{n}\left(t\right)\right]$,
as the limit of the expected instantaneous throughput over time. In
a single operator network with $N$ users, as formally proposed by
\cite{kelly1997}, proportional fairness is achieved by a feasible
throughput allocation $\left\{ \bar{\mu}^{*}_{n},n=1,2,\ldots,N\right\} $
such that for any other feasible allocation $\left\{ \bar{\mu}_{n},n=1,2,\ldots,N\right\} $,
the sum of proportional changes satisfies:
\begin{equation}
\sum^{N}_{n=1}\left(\bar{\mu}_{n}-\bar{\mu}^{*}_{n}\right)/\bar{\mu}^{*}_{n}\leq0.
\end{equation}
This definition implies that any unilateral improvement in a user\textquoteright s
throughput must be offset by a disproportionate reduction in others\textquoteright{}
allocations, thereby discouraging spectral monopolization. The intrinsic
value of proportional fairness lies in its dual emphasis on efficiency
and equity, ensuring that users with weaker channel conditions are
not perpetually marginalized while maintaining high aggregate throughput. 

The relationship between proportional fairness and optimization theory
is established through resource allocation problems. As proven in
\cite{kelly1997}, the solution inherently satisfies the proportional
fairness criterion with the logarithmic utility function. In this
case, the PFS corresponds to solving the following optimization problem:
\begin{equation}
\textrm{maximize}\sum^{N}_{n=1}\ln\left(\bar{\mu}_{n}\right),
\end{equation}
which is subject to system feasibility constraints. The logarithmic
utility function inherently prioritizes fairness by assigning diminishing
marginal returns to throughput increases. This property ensures that
modest improvements for users with historically low throughput yield
higher utility than marginal gains for users already operating at
high rates. 

In a single-operator network, the PFS operates as follows. The operator
acquires users' instantaneous rates $\ensuremath{\left\{ r_{n},n=1,2,\ldots,N\right\} }$
through feedback mechanisms, where users periodically report their
channel state information (CSI). The operator then selects the user
maximizing the ratio of the transmission rates to throughput $\ensuremath{\left\{ \mu_{n},n=1,2,\ldots,N\right\} }$.
Mathematically, user selection follows proportional fair rule:
\begin{equation}
n^{*}\left(t+1\right)=\underset{1\leq n\leq N}{\text{argmax}}\frac{r_{n}\left(t+1\right)}{\mu_{n}\left(t\right)}.\label{eq:=00591A=009891=006BB5=008C03=005EA6=0089C4=005219}
\end{equation}
After user selection, the operator allocates spectrum resources and
configures modulation schemes to maximize the selected user's data
rate. Data transmission starts with dynamic parameter adjustments
based on real-time CSI updates. Simultaneously, the selected user's
average throughput updates to incorporate the current rate via:
\begin{equation}
\mu_{n}\left(t+1\right)=\left(1-\alpha\right)\mu_{n}\left(t\right)+\alpha r_{n}\left(t+1\right)I_{n}\left(t+1\right).\label{eq:=00591A=009891=006BB5=00541E=005410=0091CF=0066F4=0065B0}
\end{equation}
The smoothing factor $\alpha\in\left(0,1\right)$ controls the throughput
adaptation rate. Smaller $\alpha$ yields gradual throughput adjustments,
while larger $\alpha$ enables faster responses to channel fluctuations.

The basic goal of PFS is to dynamically adapt resource allocation
over time, ensuring that long-term average throughput converges to
the optimal solution. By integrating real-time channel state information
with historical throughput data, these algorithms strike a balance
between exploiting favorable channel conditions and achieving fairness.

However, while proved effective in single-operator environments, the
PFS faces inevitable challenges in multi-operator O-RAN due to the
absence of trust. In O-RAN, coordinated user scheduling across operators
is significant to enable efficient resource utilization. Serving as
the foundational prerequisite for such cooperation, trust among operators
ensures reliable commitment to scheduling patterns and prevents deviations
from agreed policies. Without robust trusted multi-operator collaboration
mechanisms, scheduling policies in existing approaches may fail in
practice due to opportunistic or malicious behaviors like free-riding
or strategic non-compliance, which result in globally suboptimal performance.
Therefore, the gap in current research highlights the critical need
for novel scheduling paradigms that enable trustworthy multi-operator
collaboration within the open and disaggregated architecture of O-RAN. 

\section{Blockchain-based Proportional Fair Scheduling\label{sec:Blockchain-based-PF-Scheduling}}

\subsection{Overview}

To overcome these limitations of traditional user scheduling in Section
\ref{sec:Revisit-PF-Scheduling}, we integrate blockchain into the
PFS framework and design the BC-PFS to bridge the trust gap in multi-operator
collaboration and preserve the fairness-efficiency balance. This section
presents the BC-PFS scheme for trustworthy and efficient user association
in multi-operator O-RAN.

To support the proposed BC-PFS, we introduce a two-layer blockchain
architecture into O-RAN, consisting of multiple edge chains and a
global chain (See Fig. \ref{fig:The-system-architecture}). The edge
chain is deployed at the edge O-Cloud and corresponds to a sub-region,
defined as a geographical service area covered by a group of O-DUs
from multiple operators. It coordinates cross-operator user scheduling
within a sub-region. The global chain is deployed in proximity to
the SMO and aggregates information from edge chains to support overall
network management across sub-regions. 

For user scheduling, each O-DU collects user states and submits the
required scheduling information to the edge chain. Based on the reported
user states and historical throughput, the edge chain executes PFS
algorithm via smart contracts and delivers the scheduling decisions
to the corresponding O-DUs, which then translate these decisions into
radio resource allocation for transmission. 

Meanwhile, the edge chain periodically uploads aggregated state information
and historical transaction records to the global chain, providing
consistent global state view and auditability. The global chain is
responsible for transaction settlement, spectrum management, and network
coordination. Moreover, based on the scheduling information maintained
by global chain, the SMO and Non-RT RIC can further perform long-term
policy optimization.

In principle, BC-PFS is the collaborative scheduling of multiple operators.
Specifically, each operator independently schedules one user per time
slot. Operator $k$ schedules user $n^{*}_{k}$ at slot $t+1$ according
to the proportional fair rule:
\begin{equation}
n^{*}_{k}\left(t+1\right)=\underset{1\leq n\leq KN}{\text{argmax}}\frac{r_{nk}\left(t+1\right)}{\mu_{n}\left(t\right)},\:\textrm{for}\:k=1,\ldots,K.\label{eq:scheduling criteria}
\end{equation}
After user selection, the throughput is updated iteratively to reflect
the latest resource allocations:
\begin{equation}
\mu_{n}\left(t+1\right)=\left(1-\alpha\right)\mu_{n}\left(t\right)+\alpha\sum^{K}_{k=1}r_{nk}\left(t+1\right)I_{nk}\left(t+1\right).\label{eq:eq11}
\end{equation}
Remark that this approach can achieve stationary average throughput
by dynamically adapting to channel variations and adjusting scheduling
decisions over time. By integrating instantaneous channel conditions
with long-term performance, it ensures that the average throughput
converges to the globally optimal solution of the maximization problem
\cite{Kushner2004}:
\begin{equation}
\textrm{maximize}\sum^{KN}_{n=1}\ln\left(\bar{\mu}_{n}\right).\label{eq:9}
\end{equation}
In the proposed BC-PFS, the blockchain serves as a core to coordinate
cross-operator scheduling for O-RAN. First, it provides a transparent
platform for inter-operator spectrum pooling in O-RAN, enhancing the
capability of multiple operators to trustfully share and aggregate
licensed frequency bands across networks. Second, the edge chain collects
authenticated user states, executes the PFS rule through smart contracts,
and records scheduling decisions immutably, ensuring verifiable cross-operator
user scheduling in a sub-region. Third, the global chain aggregates
service and transaction records from edge chains to support overall
coordination, service settlement, and network management. Therefore,
blockchain enables cross-network resource pooling and trustworthy
collaborative scheduling. By ensuring both transparent execution and
traceable records of scheduling, it establishes a reliable trust foundation
for inter-operator collaboration. Enhanced by blockchain, BC-PFS ensures
that all participants operate under verifiable protocols through automated
execution of smart contracts, while enforcing agreement on real-time
scheduling decisions via consensus mechanisms. 
\begin{figure}
\centering\includegraphics[width=0.4\paperwidth]{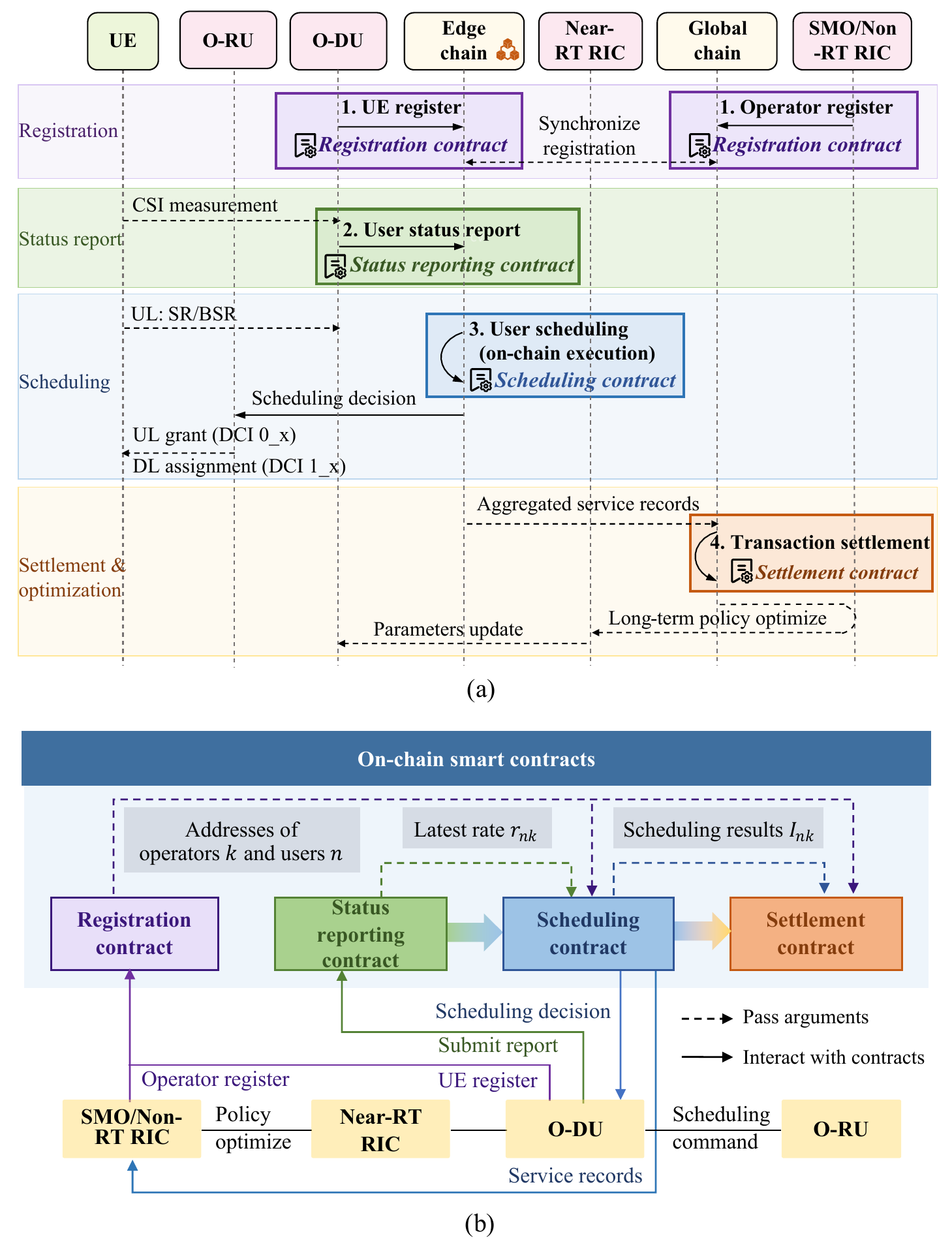}\caption{\foreignlanguage{american}{Implementation of the BC-PFS in O-RAN. (a) Brief workflow of the BC-PFS.
(b) Interactions among smart contracts and network entities. \label{fig:System-workflow-of}}}

\vspace{-2mm}
\end{figure}

\subsection{Workflow\label{subsec:Workflow}}

In this subsection, we present the complete cycle in Fig. \ref{fig:System-workflow-of}(a)
to illustrate the proposed BC-PFS in O-RAN. 

1) User and operator registration. The workflow initiates with user
and operator registration. After user accesses the network, the O-DU
collects user's registration information and submits it to the edge
chain. The registration contract verifies the user identity and records
the authenticated identities on-chain. Meanwhile, each operator registers
with the global chain through the SMO. The global chain and edge chains
then synchronize necessary registration information.

2) User status report. Following registration, users continuously
measure their channel conditions and periodically upload status reports
including critical metrics such as CSI. The O-DU preprocesses these
reports and submits required scheduling information to the edge chain.
The status reporting contract aggregates and validates reports and
updates the latest user states stored on-chain, ensuring scheduling
decisions are based on up-to-date channel conditions. 

3) User scheduling. UE with uplink (UL) transmission demand sends
scheduling request (SR) and buffer status report (BSR) to the O-DU
for indication, while downlink (DL) transmission demands are available
at the network side. During each scheduling interval, the scheduling
contract retrieves the latest user states, and queries the historical
throughput records from the previous validated block. It then executes
the PFS algorithm using the instantaneous rates and historical throughputs
to determine the scheduled users and update throughput records. The
scheduling decision is delivered to the corresponding O-DU, which
performs the resource allocation for the selected users and conveys
the resulting UL grant and DL assignment to the UEs through downlink
control information (DCI) formats 0\_x and 1\_x, respectively.

4) Settlement and optimization. Upon scheduling, the edge chain aggregates
service records and periodically synchronizes them with the global
chain, where the settlement contract calculates cross-operator service
charges according to predefined pricing and settlement rules. All
settlement records, including billing details and service quality
metrics, are permanently recorded on-chain, forming verifiable network
service credentials. Meanwhile, the Non-RT RIC analyzes long-term
service data from the global chain to optimize network policies and
deliver to the Near-RT RIC, which further optimizes scheduling parameters
for execution by O-DU. 

The above process enables trustworthy user scheduling across network
operators through secure verification, immutable records, and traceable
auditing. In particular, BC-PFS is built on the permissioned blockchain
jointly maintained by the participating operators, and can use the
consensus protocols such as Raft \cite{Luo2024}. This design is suitable
for the multi-operator scenario, since cross-operator coordination
requires controlled membership. Particularly, the consensus process
based on Raft incurs relatively low energy consumption and is favorable
for resource-constrained wireless deployments. Cryptographic techniques,
particularly digital signatures, can ensure the integrity and authenticity
of all status reports and transactions, preventing unauthorized modifications,
and enabling secure verification of participants\textquoteright{}
identities throughout the scheduling process.

Smart contract can automatically enforce scheduling decisions and
payments. Thanks to the low complexity of the PFS, the corresponding
smart contracts can be easily implemented based on users\textquoteright{}
instantaneous rates and historical throughputs in the current and
previous blocks. The online incremental update style of the PFS also
fits well with on-chain operations via smart contracts since the blockchain
information is also updated incrementally. Moreover, the scheduling
operations and transactions are immutably recorded on the blockchain,
which ensures transparency and establishes trust between different
networks.

Note that the BC-PFS requires the consensus among different operators,
which inevitably introduces extra delay. More specifically, the scheduling
interval is determined by the block time, i.e., the average time to
create a new block which contains the most recent scheduling result.
If the channel coherence time is much shorter than the block time,
or equivalently the scheduling interval, the BC-PFS will rely on outdated
channel states, which could largely degrade the performance. Therefore,
permissioned ledgers with short block time are more preferred over
public permissionless blockchains like Bitcoin whose block time is
around 10 minutes. For instance, Kaspa \cite{kaspa2025} achieves
the block time on the order of 100 ms, and Quorum \cite{Quorum2018}
can be configured to generate blocks every 50 ms. Furthermore, the
project of MegaEth \cite{megaeth2025} can achieve the block time
of 10 ms and further claims it can be reduced to around 1 ms, approaching
real-time operations. These low-latency blockchains can achieve the
block time close to or even shorter than the channel coherence time
in scenarios with slow or moderate mobility (10-50 ms) \cite{Jung2011}.
(As a reference, a pedestrian walking at 3 km/h has the coherence
time around 76 ms when the carrier frequency 2 GHz is used.) Admittedly,
in high-mobility cases, the consensus process may have a non-negligible
impact on the performance of BC-PFS. 

\subsection{Implementation via Smart Contracts\label{subsec:implementation}}

In this subsection, we present the realization of BC-PFS via smart
contracts. Fig. \ref{fig:System-workflow-of}(b) depicts the interactions
among smart contracts and O-RAN entities, where smart contracts execute
the core business logic of BC-PFS by handling registration, user status
reporting, scheduling, and settlement. The design and functionality
of four core smart contracts are as follows. Furthermore, we provide
Solidity codes in terms of the four core smart contracts in Algorithms
1-4, respectively. \footnote{The repository address of the source code is given in the Introduction.}
\begin{itemize}
\item The registration contract (see Algorithm 1) processes both user and
operator registration requests. During registration, the contract
first verifies two critical conditions: the entity ID must be unregistered
and the provided credentials must be validated according to the validProof
function. For successful registrations, the contract updates the registry
mapping with a tuple containing the generated ID and role flag, then
emits a RegistrationSuccess event containing the address, assigned
ID, and role type.\begin{algorithm}[t] \caption{Registration Contract} \begin{algorithmic}[1] \setlength{\algorithmicindent}{0.5em} \raggedright
\STATE struct Info \{ 
\STATE \hspace*{0.5em}uint id; \STATE \hspace*{0.5em}bool isOperator; \} 
\STATE uint private nextOpId = 1; \STATE uint private nextUserId = 1; 
\STATE mapping(address => Info) public registry; 
\STATE event RegistrationSuccess(address user, uint id, bool isOperator); 
\STATE \textbf{function register}(address user, bytes memory proof, bool isOperator) public \{ \setlength{\algorithmicindent}{0.5em}
\STATE \hspace*{0.5em}\textit{//  Identity validation} 
\STATE \hspace*{0.5em}require(validProof(proof, isOperator), "Invalid 
\hspace*{0.5em}credentials"); 
\STATE \hspace*{0.5em}\textit{// Registration} 
\STATE \hspace*{0.5em}if (isOperator) \{ \STATE \hspace*{1em}registry[user] = Info(nextOpId++, true); 
\STATE \hspace*{0.5em}\} else \{ 
\STATE \hspace*{1em}registry[user] = Info(nextUserId++, false); \} 
\STATE \hspace*{0.5em}emit RegistrationSuccess(user, registry[user].id, 
\hspace*{0.5em}isOperator); \} \end{algorithmic} \end{algorithm}
\end{itemize}
\begin{algorithm}[t] 
\caption{Status Reporting Contract} 
\begin{algorithmic}[1] 
\setlength{\algorithmicindent}{0.5em}  
\raggedright
\STATE mapping(address => mapping(address => uint[])) public operatorUserRates; 
\STATE event ReportSubmitted(address user, address operator, uint timestamp, uint rate); 
\STATE \textbf{function submitReport}(address user, address operator, bytes memory csi) public \{  
\STATE \hspace*{0.5em}\textit{// CSI validation}  \STATE \hspace*{0.5em}require(validCSI(csi), "Invalid CSI data");  
\STATE \hspace*{0.5em}\textit{// Rate estimation}  \STATE \hspace*{0.5em}uint rate = rateEstimation(csi);  
\STATE \hspace*{0.5em}\textit{// Rate update}  \STATE \hspace*{0.5em}operatorUserRates[user][operator].push(rate);  
\STATE \hspace*{0.5em}emit ReportSubmitted(user, operator, block.timestamp, \hspace*{0.5em}rate);   \} \end{algorithmic} \end{algorithm}

\begin{algorithm}[t] \caption{Scheduling Contract} \begin{algorithmic}[1] \setlength{\algorithmicindent}{0.5em}  \raggedright
\STATE mapping(address => address) public selectedUser; 
\STATE mapping(address => uint) public throughput; 
\STATE uint public constant alphaInv = 10000; 
\STATE mapping(address => mapping(address => bool)) public schedulingMatrix; 
\STATE mapping(address => mapping(address => uint)) public allocatedRate;
\STATE event Scheduled(address[] operators, address[] selectedUser);
\STATE \textbf{function updateScheduling}() public \{ 
\STATE \hspace*{0.5em} \textit{// User selection via (8)}    
\STATE \hspace*{0.5em}address[] memory users = getAllUsers(); 
\STATE \hspace*{0.5em}address[] memory operators = getAllOperators(); 
\STATE \hspace*{0.5em}for (uint k = 0; k < operators.length; k++) \{ 
\STATE \hspace*{1em}address op = operators[k];
\STATE \hspace*{1em}address prev = selectedUser[op];
\STATE \hspace*{1em}schedulingMatrix[prev][op] = false;
\STATE \hspace*{1em}allocatedRate[prev][op] = 0;
\STATE \hspace*{1em}uint maxPriority = 0; 
\STATE \hspace*{1em}uint bestLatestRate = 0;
\STATE \hspace*{1em}selectedUser[op] = address(0); 
\STATE \hspace*{1em}for (uint n = 0; n < users.length; n++) \{ 
\STATE \hspace*{1.5em}uint latestRate =\\
\hspace*{1.5em}StatusReporting.getLatestRate(users[n], op);
\STATE \hspace*{1.5em}uint priority = latestRate / throughput[users[n]]; 
\STATE \hspace*{1.5em}if (priority > maxPriority) \{ 
\STATE \hspace*{2em}maxPriority = priority; 
\STATE \hspace*{2em}selectedUser[op] = users[n];
\STATE \hspace*{2em}bestLatestRate = latestRate;\}\} 
\STATE \hspace*{1em}schedulingMatrix[selectedUser[op]][op] = true;
\STATE \hspace*{1em}allocatedRate[selectedUser[op]][op] = bestLatestRate;\}
\STATE \hspace*{0.5em} \textit{// Throughput update via (9)}     
\STATE \hspace*{0.5em}for (uint n = 0; n < users.length; n++) \{ 
\STATE \hspace*{1em}uint totalAllocated = 0; 
\STATE \hspace*{1em}for (uint k = 0; k < operators.length; k++) \{ 
\STATE \hspace*{1.5em}address op = operators[k]; 
\STATE \hspace*{1.5em}if (schedulingMatrix[users[n]][op]) \{ 
\STATE \hspace*{2em}notify(users[n], operators[k]); 
\STATE \hspace*{2em}totalAllocated += allocatedRate[users[n]][op];\}\}
\STATE \hspace*{1em}throughput[users[n]] = ((alphaInv - 1) * \hspace*{1em}throughput[users[n]] + totalAllocated) / alphaInv; \}
\STATE \hspace*{0.5em}emit Scheduled(operators, selectedUser); \} \end{algorithmic} \end{algorithm}

\begin{algorithm}[t] \caption{Settlement Contract} \begin{algorithmic}[1] \setlength{\algorithmicindent}{0.5em}  \raggedright
\STATE event PaymentProcessed(address user, address operator, uint cost);
\STATE \textit{// Notification of scheduling details} 
\STATE \textbf{function processSettlement}() public \{ 
\STATE \hspace*{0.5em}address[] memory users = getAllUsers(); 
\STATE \hspace*{0.5em}address[] memory operators = getAllOperators(); 
\STATE \hspace*{0.5em}for (uint n = 0; n < users.length; n++) \{ 
\STATE \hspace*{1em}for (uint k = 0; k < operators.length; k++) \{ 
\STATE \hspace*{1.5em}bool scheduled = Scheduling.getSchedulingMatrix\\
\hspace*{1.5em}(users[n], operators[k]);
\STATE \hspace*{1.5em}if (scheduled) \{ 
\STATE \hspace*{2em}(uint duration, uint bandwidth) = \hspace*{2em}getServiceParameters(users[n], operators[k]); 
\STATE \hspace*{2em}settleService(users[n], operators[k], duration, \hspace*{2em}bandwidth); \}\}\}\}
\STATE \textit{// Service settlement} 
\STATE \textbf{function settleService}(address user, address operator, uint duration, uint bandwidth) public \{ 
\STATE \hspace*{0.5em}uint cost = calculateCost(operator, duration, bandwidth); 
\STATE \hspace*{0.5em}processPayment(user, operator, cost); 
\STATE \hspace*{0.5em}emit PaymentProcessed(user, operator, cost); \}\} 
 \end{algorithmic} \end{algorithm}
\begin{itemize}
\item The status reporting contract (see Algorithm 2) processes real-time
channel state updates from registered users. The contract validates
incoming CSI through the validCSI function before processing, rejecting
reports with invalid data formats. Approved reports undergo rate estimation
through the rateEstimation function before being permanently recorded
in the operatorUserRates mapping, with successful updates triggering
ReportSubmitted events.
\item The scheduling contract (see Algorithm 3) implements the BC-PFS algorithm
to execute user scheduling. The contract is initiated by evaluating
all possible user-operator pairs through a two-stage process. First,
it calculates priority scores for each pair by dividing the user's
instantaneous transmission rate by their historical throughput average,
simultaneously updating the schedulingMatrix. Each operator then selects
the highest-priority user via \eqref{eq:scheduling criteria}. Second,
each user's throughput record is updated with newly allocated rates
according to \eqref{eq:eq11}. Upon selection, the contract notifies
participants of scheduling results and emits Scheduled events to provide
an immutable on-chain record of all allocation decisions.
\item The settlement contract (see Algorithm 4) performs financial settlement
based on the scheduling results. The process is initiated through
the processSettlement function, which systematically scans the schedulingMatrix
to identify active allocations. Each validated scheduling transaction
triggers the core settleService function, which dynamically computes
service costs using operator-specific pricing metrics in calculateCost.
The contract automatically initiates fund transfers via processPayment
while broadcasting transaction confirmation through PaymentProcessed
events.
\end{itemize}

\section{Performance Analysis\label{sec:Performance-Analysis}}

\subsection{Analysis on Throughput\label{subsec:User-average-throughput}}

In this section, we would like to quantitatively evaluate the performance
of the proposed BC-PFS in O-RAN through rigorous mathematical analysis.
The PFS framework exhibits well-established weak convergence properties
\cite{Zhou2011}. The following lemma characterizes the asymptotic
behavior of PFS:
\begin{lem}
Under stationary user rate conditions, we define $\mu$ and $r$ as
the canonical representations of $\mu\left(t\right)$ and $r\left(t\right),$
respectively. The user throughput $\mu_{nk}\left(t\right)$ generated
by the scheduling algorithm in \eqref{eq:scheduling criteria} and
\eqref{eq:eq11} weakly converges to the limit point of the following
ODE system.\label{lem:=00591A=009891=006BB5ODE}
\begin{equation}
\dot{\mu}_{nk}=\bar{h}_{nk}\left(\boldsymbol{\mu}\right)-\mu_{nk},\:n=1,...,KN,k=1,...,K,
\end{equation}
where $\boldsymbol{\mu}=\left\{ \mu_{nk},\:n=1,...,KN,k=1,...,K\right\} $
and $\bar{h}_{nk}\left(\boldsymbol{\mu}\right)$ represents the stationary
expectation of the rate provided by operator $k$ to user $n$ under
the event $\frac{r_{nk}}{\mu_{n}}>\frac{r_{mk}}{\mu_{m}},\forall m\neq n$:
\begin{align}
\bar{h}_{nk}\left(\boldsymbol{\mu}\right)= & E\left[r_{nk}\mid\frac{r_{nk}}{\mu_{n}}>\frac{r_{mk}}{\mu_{m}},\forall m\neq n,m\leq KN\right].
\end{align}
\end{lem}
Notably, the canonical representations $\mu$ and $r$ characterize
the time-independent essential features of these quantities by abstracting
away their temporal dependencies. This convergence result provides
fundamental insights. First, the weak convergence property demonstrates
that the stochastic BC-PFS process ultimately follows the deterministic
ODE solution, with the distribution of $\mu_{nk}\left(t\right)$ converging
to the ODE's trajectory despite short-term random fluctuations. Therefore,
the long-term average throughput $\bar{\mu}_{nk}$ equals the ODE's
limit point and $\mu_{nk}\left(t\right)$ converges to $\bar{\mu}_{nk}$
when $t\rightarrow\infty$. Second, the ODE system admits a unique
limit point, independent of initial conditions, which implies the
uniqueness of $\bar{\mu}_{nk}$. This convergence enables our subsequent
steady-state analysis.

From Lemma \ref{lem:=00591A=009891=006BB5ODE}, we take expectations
on both sides of \eqref{eq:eq11}, which yields
\begin{align}
 & E\left[\mu_{n}\left(t+1\right)\right]=\nonumber \\
 & \left(1-\alpha\right)E\left[\mu_{n}\left(t\right)\right]+\alpha\sum^{K}_{k=1}E\left[r_{nk}\left(t+1\right)I_{nk}\left(t+1\right)\right].
\end{align}
At the long-term steady state, this simplifies to
\begin{align}
\bar{\mu}_{n}= & \underset{t\rightarrow\infty}{\text{lim}}E\left[\mu_{n}\left(t+1\right)\right]\nonumber \\
= & \underset{t\rightarrow\infty}{\text{lim}}\sum^{K}_{k=1}E\left[r_{nk}\left(t+1\right)I_{nk}\left(t+1\right)\right].\label{eq:12}
\end{align}
By using Bayes' theorem, the average throughput can be expanded as:
\begin{align}
\bar{\mu}_{n}= & \underset{t\rightarrow\infty}{\text{lim}}\sum^{K}_{k=1}\int^{\infty}_{0}\left(xf_{r_{nk}}\left(x\right)\right.\nonumber \\
 & \left.\cdot\textrm{Pr}\left(I_{nk}\left(t+1\right)=1\left|r_{nk}\left(t+1\right)=x\right.\right)\right)dx.\label{eq:16}
\end{align}
Moreover, under the scheduling rule \eqref{eq:scheduling criteria},
we have:
\begin{align}
 & \bar{\mu}_{n}=\label{eq:18}\\
 & \underset{t\rightarrow\infty}{\text{lim}}\sum^{K}_{k=1}\int^{\infty}_{0}xf_{r_{nk}}\left(x\right)\prod^{KN}_{\substack{m=1,\\
m\ne n
}
}\textrm{Pr}\left(\frac{x}{\mu_{n}\left(t\right)}>\frac{r_{mk}\left(t+1\right)}{\mu_{m}\left(t\right)}\right)dx.\nonumber 
\end{align}
By taking the limit, \eqref{eq:18} can be written as:
\begin{equation}
\bar{\mu}_{n}=\sum^{K}_{k=1}\int^{\infty}_{0}xf_{r_{nk}}\left(x\right)\prod^{KN}_{\substack{m=1,\\
m\ne n
}
}F_{r_{mk}}\left(\frac{\bar{\mu}_{m}}{\bar{\mu}_{n}}x\right)dx,\label{eq:33}
\end{equation}
where $f_{r}\left(x\right)$ and $F_{r}\left(x\right)$ represent
the probability density function and the cumulative distribution function
of user rate $r$, respectively. The mathematical structure of \eqref{eq:33}
reveals the intricate relationship between average throughput and
instantaneous rates in BC-PFS. For Rayleigh fading channels, \cite{Liu2011}
established performance bounds under proportional fairness as follows.
\begin{lem}
\label{lem:=005927=005C0F=005173=007CFB}In Rayleigh fading environments,
given two users $i$ and $j$ with $\bar{r}_{i}\leq\bar{r}_{j}$,
under the definition of proportional fairness, $\frac{\sigma_{j}}{\sigma_{i}}\leq\frac{\bar{\mu}_{j}}{\bar{\mu}_{i}}\leq\frac{\bar{r}_{j}}{\bar{r}_{i}}$,
where $\bar{\mu}_{i}$ and $\bar{\mu}_{j}$ represent the average
throughput of users $i$ and $j$, respectively.
\end{lem}
These bounds demonstrate that average throughput ratios are constrained
by both average rate differences and channel variability. Lemma \ref{lem:=005927=005C0F=005173=007CFB}
can be extended to multi-operator scenarios. We assume that user $n$'s
average rate $\bar{r}_{nk}$ provided by operator $k$ is mainly determined
by the user itself. This assumption is realistic when the large-scale
fading depending on the user's location is dominant over the channel
quality. That is, the transmission rates provided by different operators
$\bar{r}_{nk}$ exhibit minor fluctuations around user $n$'s average
rate over $K$ networks $\frac{\bar{r}_{n}}{K}$. According to our
assumption and Lemma \ref{lem:=005927=005C0F=005173=007CFB}, under
the linear rate approximation model over Rayleigh fading, the relationship
between the average throughput of different users satisfies
\begin{equation}
\frac{\bar{\mu}_{n}}{\bar{\mu}_{m}}=\frac{\bar{r}_{n}}{\bar{r}_{m}}.\label{eq:equal}
\end{equation}
This equality provides a theoretical foundation for deriving \eqref{eq:33},
and reduces the computational complexity of solving inter-network
scheduling optimization problems by directly relating average throughput
ratios to channel rate. For user $n$, the average throughput $\bar{\mu}_{n}$
expressed in \eqref{eq:33} can be rewritten as:
\begin{align}
 & \bar{\mu}_{n}=\label{eq:15}\\
 & \sum^{K}_{k=1}\int^{\infty}_{0}\frac{\textrm{exp}\left(-\frac{x}{\bar{r}_{nk}}\right)x}{\bar{r}_{nk}}\prod^{KN}_{\substack{m=1,\\
m\ne n
}
}\left(1-\textrm{exp}\left(-\frac{x\bar{r}_{m}}{\bar{r}_{mk}\bar{r}_{n}}\right)\right)dx.\nonumber 
\end{align}

The result in \eqref{eq:15} provides a semi-closed-form solution
to describe the average throughput for the BC-PFS in O-RAN. We would
like to further expand the integral into a summation form via the
inclusion-exclusion principle, yielding a more precise expression
at the cost of higher computational complexity:
\begin{align}
\bar{\mu}_{n}= & \sum^{K}_{k=1}\frac{1}{\bar{r}_{nk}}\int^{\infty}_{0}\left(\textrm{exp}\left(-\frac{x}{\bar{r}_{nk}}\right)x\sum_{\mathcal{S}\in\mathcal{P}\left(\mathcal{A}_{n}\right)}\left(-1\right)^{\left|\mathcal{S}\right|}\right.\nonumber \\
 & \left.\cdot\textrm{exp}\left(-x\sum_{m\in\mathcal{S}}\frac{\bar{r}_{m}}{\bar{r}_{mk}\bar{r}_{n}}\right)\right)dx,
\end{align}
where $\mathcal{A}_{n}=\left\{ 1,...,KN\right\} \setminus\left\{ n\right\} $
denotes the complete set of all user indices excluding user $n$ and
$\mathcal{P}\left(\mathcal{A}_{n}\right)$ represents its power set
containing all possible subsets.\footnote{For instance, $\mathcal{P}\left(\left\{ 1,2\right\} \right)=$$\left\{ \emptyset,\left\{ 1\right\} ,\left\{ 2\right\} ,\left\{ 1,2\right\} \right\} $.}
Meanwhile, $\mathcal{S}$ is an arbitrary subset from an element in
$\mathcal{P}\left(\mathcal{A}_{n}\right)$, which is a user subset,
and $\left|\mathcal{S}\right|$ indicates the number of elements in
$\mathcal{S}$. Upon substituting the expanded summation form into
the integral, the average throughput is derived as:
\begin{align}
\bar{\mu}_{n}= & \sum^{K}_{k=1}\sum_{\mathcal{S}\in\mathcal{P}\left(\mathcal{A}_{n}\right)}\frac{\left(-1\right)^{\left|\mathcal{S}\right|}}{\bar{r}_{nk}}\left(\frac{1}{\bar{r}_{nk}}+\sum_{m\in\mathcal{S}}\frac{\bar{r}_{m}}{\bar{r}_{mk}\bar{r}_{n}}\right)^{-2}.\label{eq:solution1}
\end{align}
This formulation provides a complete characterization of user average
throughput in the blockchain-enhanced O-RAN, and quantifies how multi-operator
O-RAN aggregation affects system performance. The accuracy of \eqref{eq:15}
and \eqref{eq:solution1} is verified by the simulations in Section
\ref{sec:Simulation-and-analysis}. Note that for large-scale scenarios
with numerous operator networks or users per network, \eqref{eq:solution1}
may still require large computational overheads.

\subsection{ODE Approach\label{subsec:ODE-Approach}}

Remark that the above results can also be derived equivalently via
ODE. The user throughput $\mu_{nk}\left(t\right)$ in the BC-PFS weakly
converges to the limiting solution of a coupled ODE system. We will
show the derivation in the simplest non-trivial case with two operator
networks and one user per network, i.e., $K=2$ and $N=1$. The expected
instantaneous rates decompose as $\bar{r}_{1}=\bar{r}_{11}+\bar{r}_{12}$
and $\bar{r}_{2}=\bar{r}_{21}+\bar{r}_{22}$, while the total throughput
components satisfy $\mu_{1}=\mu_{11}+\mu_{12}$ and $\mu_{2}=\mu_{21}+\mu_{22}$.
Under the linear rate model assumption, by computing $\bar{h}\left(\mu\right)$
according to Lemma \ref{lem:=00591A=009891=006BB5ODE}, the ODE system
has the following form:
\[
\begin{cases}
\dot{\mu}_{11}=\bar{r}_{11}-\frac{\frac{\mu^{2}_{1}}{\bar{r}_{11}}}{\left(\frac{\mu_{1}}{\bar{r}_{11}}+\frac{\mu_{2}}{\bar{r}_{21}}\right)^{2}}-\mu_{11},\\
\dot{\mu}_{12}=\bar{r}_{12}-\frac{\frac{\mu^{2}_{1}}{\bar{r}_{12}}}{\left(\frac{\mu_{1}}{\bar{r}_{12}}+\frac{\mu_{2}}{\bar{r}_{22}}\right)^{2}}-\mu_{12},\\
\dot{\mu}_{21}=\bar{r}_{21}-\frac{\frac{\mu^{2}_{2}}{\bar{r}_{21}}}{\left(\frac{\mu_{1}}{\bar{r}_{11}}+\frac{\mu_{2}}{\bar{r}_{21}}\right)^{2}}-\mu_{21},\\
\dot{\mu}_{22}=\bar{r}_{22}-\frac{\frac{\mu^{2}_{2}}{\bar{r}_{22}}}{\left(\frac{\mu_{1}}{\bar{r}_{12}}+\frac{\mu_{2}}{\bar{r}_{22}}\right)^{2}}-\mu_{22}.
\end{cases}
\]
The equilibrium solutions can be obtained by setting $\dot{\mu}_{nk}=0$.
By using \eqref{eq:equal}, we derive the expressions for the average
throughput:
\[
\begin{cases}
\bar{\mu}_{1}=\bar{r}_{1}-\left(\frac{\frac{\bar{r}^{2}_{1}}{\bar{r}_{11}}}{\left(\frac{\bar{r}_{1}}{\bar{r}_{11}}+\frac{\bar{r}_{2}}{\bar{r}_{21}}\right)^{2}}+\frac{\frac{\bar{r}^{2}_{1}}{\bar{r}_{12}}}{\left(\frac{\bar{r}_{1}}{\bar{r}_{12}}+\frac{\bar{r}_{2}}{\bar{r}_{22}}\right)^{2}}\right),\\
\bar{\mu}_{2}=\bar{r}_{2}-\left(\frac{\frac{\bar{r}^{2}_{2}}{\bar{r}_{21}}}{\left(\frac{\bar{r}_{1}}{\bar{r}_{11}}+\frac{\bar{r}_{2}}{\bar{r}_{21}}\right)^{2}}+\frac{\frac{\bar{r}^{2}_{2}}{\bar{r}_{22}}}{\left(\frac{\bar{r}_{1}}{\bar{r}_{12}}+\frac{\bar{r}_{2}}{\bar{r}_{22}}\right)^{2}}\right).
\end{cases}
\]

The above result has the same form as \eqref{eq:solution1}. We can
generalize the derivation to $K$-operator networks by constructing
and solving corresponding ODEs. The generalization to $KN$ users
and $K$ operator networks requires constructing conditional expectations
$\bar{h}_{nk}\left(\mu\right)$ and solving the corresponding high-dimensional
ODE system. Similarly, we can obtain the universal analytical solution:
\[
\begin{cases}
\bar{\mu}_{n}=\sum^{K}_{k=1}\sum_{\mathcal{S}\in\mathcal{P}\left(\mathcal{A}_{n}\right)}\left(-1\right)^{\left|\mathcal{S}\right|}\frac{\frac{\bar{r}^{2}_{n}}{\bar{r}_{nk}}}{\left(\frac{\bar{r}_{n}}{\bar{r}_{nk}}+\sum_{m\in\mathcal{S}}\frac{\bar{r}_{m}}{\bar{r}_{mk}}\right)^{2}},\\
\bar{r}_{n}=\sum^{K}_{k=1}\bar{r}_{nk}.
\end{cases}
\]
This comprehensive solution captures several crucial aspects of the
multi-operator environment in the BC-PFS. The double summation structure
reflects the hierarchical nature of resource competition, where users
contend both within and across operator networks. Besides, the denominator
terms precisely quantify the cumulative impact from competing users,
with each $\frac{\bar{r}_{m}}{\bar{r}_{mk}}$ term representing the
normalized competitive pressure from user $m$ through operator $k$. 

Notably, this result is exactly equivalent to the result in \eqref{eq:solution1},
which establishes theoretical consistency between the probabilistic
and ODE approaches to modeling the BC-PFS in O-RAN.

\subsection{Simplified Solution\label{subsec:Simplified-Solution}}

While the approaches in Sections \ref{subsec:User-average-throughput}
and \ref{subsec:ODE-Approach} provide theoretical assessment, they
involve high-dimensional coupling terms that introduce significant
computational complexity. To develop a more tractable analytical framework,
we would like to simplify $\bar{\mu}_{nk}$ in \eqref{eq:15} for
analyzing the impact of blockchain. By applying variable substitution
$y=\frac{x}{\bar{r}_{nk}}$, $\bar{\mu}_{nk}$ can be transformed
to:
\begin{align}
 & \bar{\mu}_{nk}=\\
 & \bar{r}_{nk}\int^{\infty}_{0}y\textrm{exp}\left(-y\right)\prod^{KN}_{m=1,m\neq n}\left(1-\textrm{exp}\left(-\frac{\bar{r}_{nk}\bar{r}_{m}}{\bar{r}_{mk}\bar{r}_{n}}y\right)\right)dy.\nonumber 
\end{align}
According to the above assumption, we denote $\nu_{n}=\frac{1}{K}\bar{r}_{n}$
as user $n$'s average rate over $K$ networks, and the transmission
rates provided by different operators $\bar{r}_{nk}$ exhibit minor
fluctuations around $\nu_{n}$, i.e., $|\frac{\bar{r}_{nk}-\nu_{n}}{\nu_{n}}|\leq\delta$,
where $\delta$ is the deviation of user rate from the average. As
a result, the average throughput $\bar{\mu}_{nk}$ for user $n$ can
be simplified to $\bar{\mu}^{\textrm{s}}_{nk}$, given by:
\begin{align}
\bar{\mu}_{nk}\approx & \nu_{n}\int^{\infty}_{0}y\textrm{exp}\left(-y\right)\left(1-\textrm{exp}\left(-y\right)\right)^{KN-1}dy\triangleq\bar{\mu}^{\textrm{s}}_{nk}.
\end{align}
By substituting $z=1-\textrm{exp}\left(-y\right)$ with Taylor expansion,
we can obtain the simplified solution $\bar{\mu}^{\textrm{s}}_{nk}$
as:
\begin{align}
\bar{\mu}^{\textrm{s}}_{nk}= & \nu_{n}\int^{1}_{0}\left(-\ln\left(1-z\right)\right)z^{KN-1}dz=\nu_{n}\sum^{\infty}_{i=1}\frac{1}{i\left(i+KN\right)}\nonumber \\
= & \frac{\nu_{n}}{KN}\sum^{\infty}_{i=1}\left(\frac{1}{i}-\frac{1}{i+KN}\right)=\frac{\nu_{n}}{KN}\omega\left(KN\right),
\end{align}
where $\omega\left(KN\right)\triangleq\sum^{KN}_{i=1}\frac{1}{i}$.
Consequently, we obtain the simplified average throughput of user
$n$ in the compact expression:
\begin{equation}
\bar{\mu}^{\textrm{s}}_{n}=\frac{\nu_{n}}{N}\omega\left(KN\right).\label{eq:simplified solution}
\end{equation}
This result in \eqref{eq:simplified solution} shows that the user
throughput of BC-PFS primarily depends on the average user rate $\nu_{n}$,
the number of operator networks $K$, and the number of users per
network $N$. Meanwhile, the factor $\omega\left(KN\right)$ explicitly
quantifies the advantages of resource pooling and inter-network cooperation
in O-RAN as network becomes more open, i.e., the total user number
$KN$ increases. This closed-form solution provides valuable insights
for performance analysis on the BC-PFS across multi-operator networks. 

More importantly, the above simplified user throughput serves as a
performance lower bound in fact. If we take the channel variability
among operator networks into consideration, the rate fluctuation across
different operator networks can provide extra multi-network diversity
gain, analogous to multi-user diversity gain. Consequently, the actual
performance of the BC-PFS in O-RAN will exceed the value in \eqref{eq:simplified solution}.
Our simulation results in Section \ref{sec:Simulation-and-analysis}
illustrate the close approximation of our simplified solution in \eqref{eq:simplified solution}
and also verify \eqref{eq:simplified solution} as a lower bound.

\section{Pooling Effect\label{sec:Blockchain-Pooling-Effect}}

After obtaining the simplified closed-form solution in \eqref{eq:simplified solution},
we would like to quantify the performance of BC-PFS in O-RAN through
a systematic utility analysis. Recall that the PFS is equivalent to
maximizing the network utility $U$\cite{Kushner2004}, given by
\begin{equation}
U\triangleq\sum^{KN}_{n=1}\ln\left(\bar{\mu}_{n}\right).
\end{equation}
For an isolated network where operators perform non-cooperative PFS,
the average throughput per user reduces to $\frac{\nu_{n}}{N}\omega\left(N\right)$.
Thus, the baseline network utility is
\begin{equation}
U_{\textrm{iso}}=\sum^{KN}_{n=1}\ln\left(\frac{\nu_{n}\omega\left(N\right)}{N}\right).
\end{equation}
Through BC-PFS, the average throughput per user is approximated as
$\frac{\nu_{n}}{N}\omega\left(KN\right)$. Consequently, the blockchain-enhanced
cooperative O-RAN achieves a higher network utility, given by
\begin{equation}
U_{\textrm{coop}}=\sum^{KN}_{n=1}\ln\left(\frac{\nu_{n}\omega\left(KN\right)}{N}\right).\label{eq:20}
\end{equation}
We define their gap as the pooling gain $\psi\left(N,K\right)$:
\begin{align}
\psi\left(N,K\right)\triangleq & U_{\textrm{coop}}-U_{\textrm{iso}}=KN\textrm{\ensuremath{\ln\left(\frac{\omega\left(KN\right)}{\omega\left(N\right)}\right)}.}\label{eq:21}
\end{align}
Apparently, the pooling gain $\psi\left(N,K\right)$ is determined
by the numbers of operator networks $K$ and users per network $N$.
Through harmonic number expansion, we can show that the pooling gain
is always positive, i.e.,
\begin{equation}
\psi\left(N,K\right)=KN\ln\left(1+\frac{\sum^{KN}_{i=N+1}\frac{1}{i}}{\omega\left(N\right)}\right)>0.\label{eq:cao}
\end{equation}
Mathematically, we prove that the O-RAN using BC-PFS outperforms the
isolated networks based on non-cooperative PFS. Since the above results
are based on the assumption without considering multi-network diversity,
the simplified solution in \eqref{eq:simplified solution}, in fact,
is a conservative lower bound. In other words, the actual network
utility will be further enhanced by the diversity on the network level,
amplifying the pooling gain beyond $\psi\left(N,K\right)$ in \eqref{eq:cao}.

Mathematically, the following theorem points out that the pooling
gain $\psi\left(N,K\right)$ can be enlarged by both network size
and the number of users. 
\begin{thm}
For $N=1,2,...$ and $K=2,3,...$, the pooling gain $\psi\left(N,K\right)$
is strictly monotonically increasing with respect to the number of
networks $K$ and the number of users per network $N$.\label{thm:monotonicity}
\end{thm}
\begin{proof}
We first prove the monotonicity in $K$. For fixed $N\geq1$, consider
the difference with respect to $K$:
\begin{align}
 & \psi\left(N,K+1\right)-\psi\left(N,K\right)\label{eq:34}\\
 & =\left(K+1\right)N\ln\left(\frac{\omega\left(\left(K+1\right)N\right)}{\omega\left(N\right)}\right)-KN\ln\left(\frac{\omega\left(KN\right)}{\omega\left(N\right)}\right)\nonumber \\
 & =N\left(\ln\left(\frac{\omega\left(\left(K+1\right)N\right)}{\omega\left(N\right)}\right)+K\ln\left(\frac{\omega\left(\left(K+1\right)N\right)}{\omega\left(KN\right)}\right)\right).\nonumber 
\end{align}
Since $\omega\left(\cdot\right)$ is an increasing function, both
logarithmic terms in \eqref{eq:34} are positive, which ensures that
$\psi\left(N,K+1\right)-\psi\left(N,K\right)>0$. Hence, $\psi\left(N,K\right)$
strictly increases in $K$. 

Consider the difference of $\psi\left(N,K\right)$ with respect to
$N$:
\begin{align}
 & \psi\left(N+1,K\right)-\psi\left(N,K\right)\\
 & =K\left(\left(N+1\right)\ln\left(\frac{\omega\left(K\left(N+1\right)\right)}{\omega\left(N+1\right)}\right)-N\ln\left(\frac{\omega\left(KN\right)}{\omega\left(N\right)}\right)\right)\nonumber \\
 & =K\left(\left(N+1\right)\ln\left(\frac{1+\frac{\sum^{K}_{i=1}\frac{1}{KN+i}}{\omega\left(KN\right)}}{1+\frac{1}{\left(N+1\right)\omega\left(N\right)}}\right)+\ln\left(\frac{\omega\left(KN\right)}{\omega\left(N\right)}\right)\right).\nonumber 
\end{align}
If $\frac{\sum^{K}_{i=1}\frac{1}{KN+i}}{\omega\left(KN\right)}\geq\frac{1}{\left(N+1\right)\omega\left(N\right)}$,
we can prove the positivity immediately. Otherwise, we have $\ln\left(\frac{1+a}{1+b}\right)=\ln\left(1+\frac{a-b}{1+b}\right)>a-b$
for $0\leq a<b$, and we can further obtain
\begin{align}
 & \psi\left(N+1,K\right)-\psi\left(N,K\right)\\
 & >K\left(\frac{\left(N+1\right)}{\omega\left(KN\right)}\sum^{K}_{i=1}\frac{1}{KN+i}-\frac{1}{\omega\left(N\right)}+\ln\left(\frac{\omega\left(KN\right)}{\omega\left(N\right)}\right)\right)\nonumber \\
 & >K\left(\frac{1}{\omega\left(KN\right)}-\frac{1}{\omega\left(N\right)}+\ln\left(\frac{\omega\left(KN\right)}{\omega\left(N\right)}\right)\right)>0.\nonumber 
\end{align}
Therefore, $\psi\left(N,K\right)$ is also strictly monotonically
increasing with $N$. 
\end{proof}
Theorem \ref{thm:monotonicity} indicates that the pooling gain $\psi\left(N,K\right)$
is always positive and also exhibits monotonic increase in both the
number of networks and users. It points out the positive impact of
cross-operator resource pooling of O-RAN enhanced by BC-PFS, since
BC-PFS helps O-RAN to establish the trust required for reliable coordination
across operators. Moreover, the pooling effect can be further enhanced
in a more open network with a larger size. On the one hand, the monotonicity
with respect to $N$ indicates that, as the number of users per network
increases, the gains from the BC-PFS become more pronounced. This
characteristic is mainly gained from multi-user diversity, while BC-PFS
in O-RAN further amplifies this effect through inter-network collaborative
scheduling. On the other hand, the monotonicity with respect to the
number of networks $K$ reflects the diversity on the network level.
The BC-PFS breaks through the resource limitations of isolated single
networks and enhances the openness of O-RAN.
\begin{figure}
\centering\subfigure[]{\includegraphics[width=0.33\paperwidth]{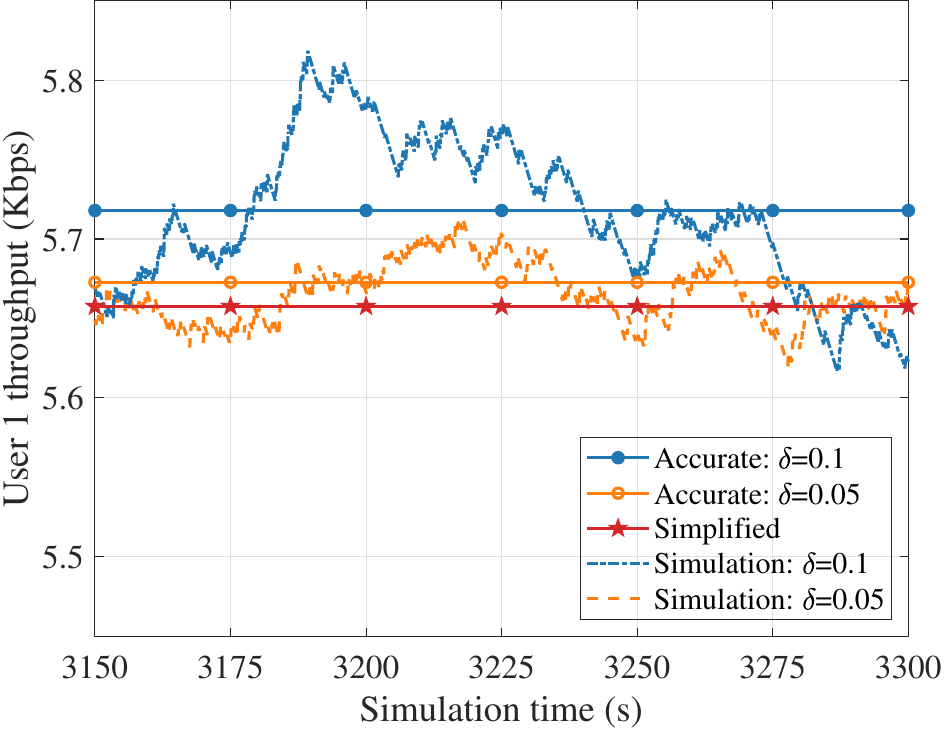}\label{N=2}}\hfil

\subfigure[]{\includegraphics[width=0.33\paperwidth]{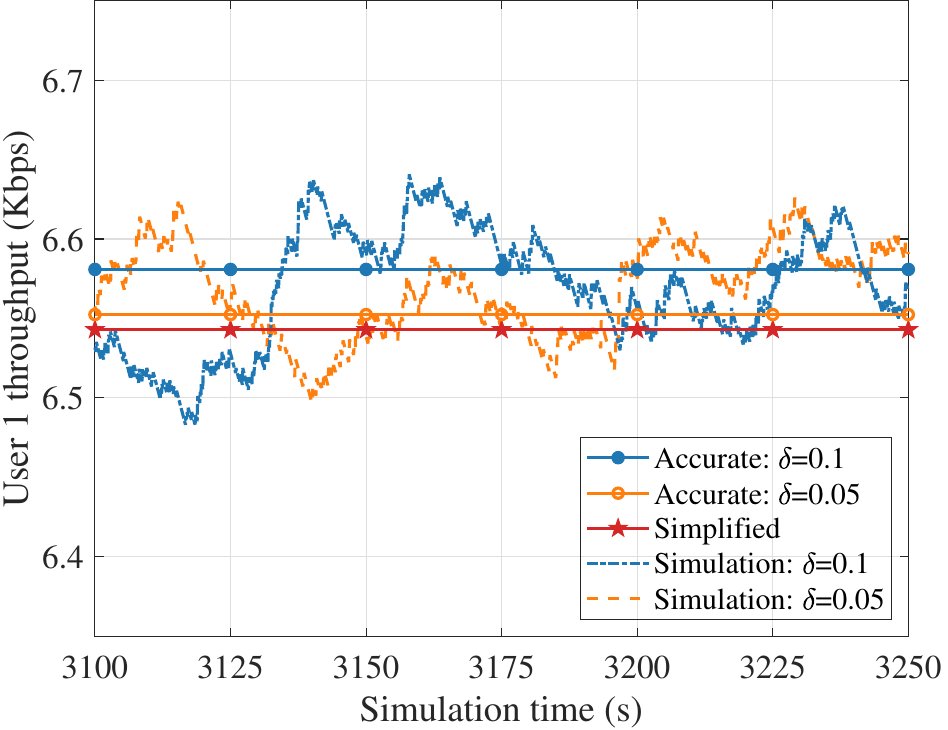}\label{N=10}}

\caption{\foreignlanguage{american}{Simulation, accurate solutions, and simplified solutions of user throughput
of the BC-PFS across $K$ operators. (a) $K=2$. (b) $K=5$.\label{fig:Simulation.-accurate-solutions.}}}
\vspace{-2mm}
\end{figure}

\section{Simulation and analysis\label{sec:Simulation-and-analysis}}

In this section, we present the simulation results to support our
analysis and conclusions. We model the wireless channel as a block-fading
Rayleigh channel with the channel coherence time of 50 ms, and the
scheduling interval is also set to 50\,ms. The throughput update window
is set to $\frac{1}{\alpha}=10000$, which ensures that the tracking
parameter $\alpha$ is sufficiently small to achieve both fair scheduling
and accurate long-term throughput measurements. We adopt the bandwidth
of 1\,MHz for every operator, with average SNR values across all users
ranging from --20\,dB to --10\,dB. 

First, we illustrate the average throughput obtained from the accurate
solution \eqref{eq:solution1}, the simplified solution \eqref{eq:simplified solution},
and the simulation results, to show the impact of the average rate
fluctuation on system performance. Fig. \ref{fig:Simulation.-accurate-solutions.}(a)
and Fig. \ref{fig:Simulation.-accurate-solutions.}(b) present a comparative
analysis of throughput performance for multi-operator networks with
configurations of $K=2$ and $K=5$, respectively. The average throughput
is evaluated under two average rate distributions $\delta=0.05$ and
$\delta=0.1$, with 10 users per network. The accurate analytical
solution in \eqref{eq:solution1} can well characterize the average
throughput from simulation, which validates the accuracy of our theoretical
analysis. Notably, for concentrated rate distributions, i.e., smaller
$\delta$, the simplified solution exhibits negligible deviation from
both simulations and the accurate solutions, confirming its practical
effectiveness under realistic conditions consistent with the assumptions
in Section \ref{subsec:Simplified-Solution}. 
\begin{figure}
\centering\includegraphics[width=0.33\paperwidth]{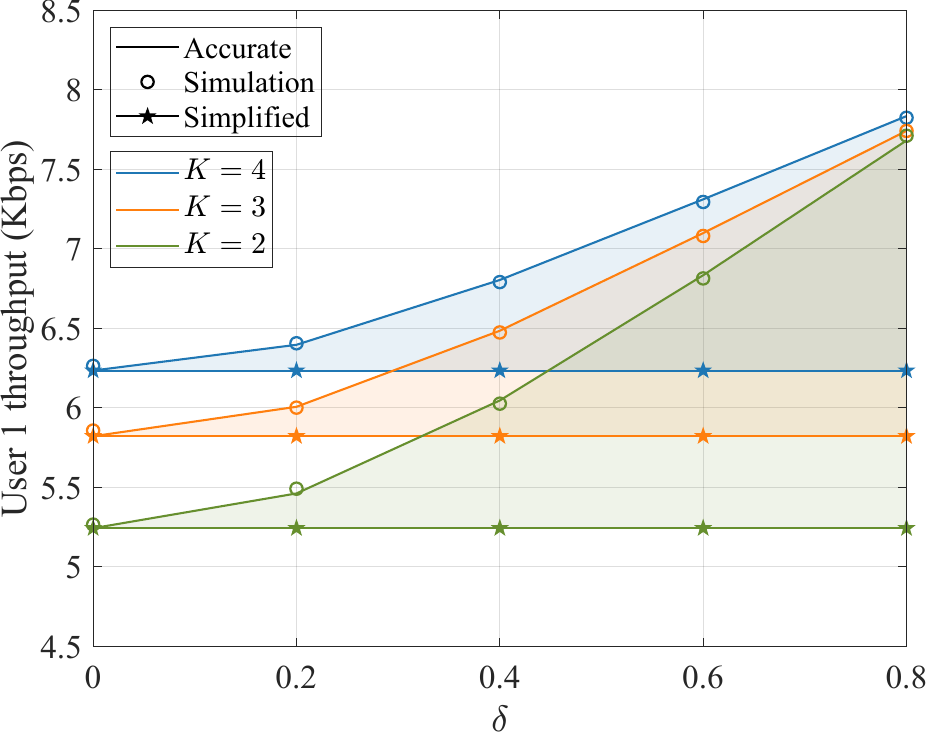}\caption{\foreignlanguage{american}{Comparison of simplified solutions, accurate solutions, and simulations
under different rate distributions.\label{fig:Comparison-of-simplified}}}
\end{figure}
\begin{figure}
\centering\includegraphics[width=0.33\paperwidth]{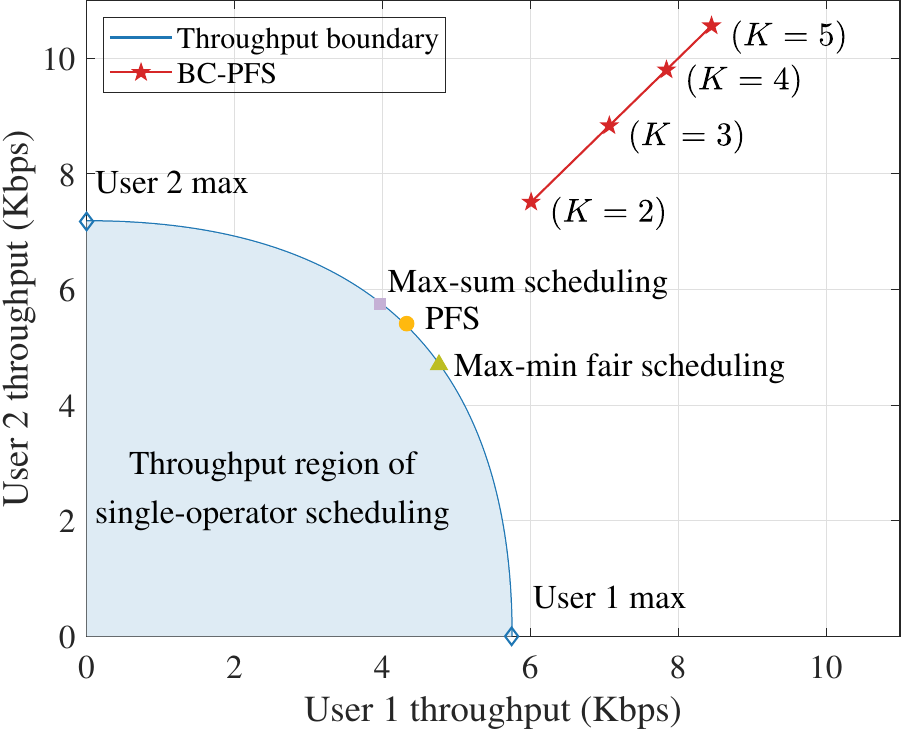}

\caption{\foreignlanguage{american}{Throughput region: case of two users.\label{fig:throughput region}}}
\end{figure}

Fig. \ref{fig:Comparison-of-simplified} illustrates the average throughput
under varying average rate distributions in multi-operator O-RAN,
where each network serves 10 users and the number of operator networks
is set to $K=2$, $K=3$, and $K=4$. The results demonstrate that
under uniform rate distribution with zero fluctuation, i.e. $\delta=0$,
the throughput reaches its lower bound, which aligns with the theoretical
conclusion presented in Section \ref{subsec:Simplified-Solution}.
As the fluctuation amplitude of the average rate increases, the user
throughput generally exhibits an upward trend, reflecting the enhanced
channel diversity gain. Notably, as the number of operator networks
increases, the deviation between the simulation and the simplified
solution diminishes for the same fluctuation amplitude. This arises
from the spatial averaging effect driven by the law of large numbers.
With more operator networks, the actual user rate is tightly clustered
around the average. Therefore, for larger values of $K$, the assumption
of small fluctuations in the average rate is more strongly met, which
improves the accuracy of the simplified solution. This consistency
underscores the robustness of the simplified solution in large-scale
O-RANs, where inter-operator coordination becomes more efficient.

Fig. \ref{fig:throughput region} compares the throughput performance
of traditional scheduling strategies and the BC-PFS for O-RAN. The
light blue area reflects the achievable region of two users' throughput
under the scenario where networks operate independently. Each user
can achieve their maximum throughput under exclusive spectrum access,
with these individual maximums defining the boundary points of the
throughput region. For traditional benchmark schemes, each operator
independently serves its own local users without cross-operator coordination.
The benchmark schemes differ in their scheduling rules. Specifically,
the PFS selects the user with the largest ratio of the instantaneous
rate to the historical average throughput, max-sum scheduling selects
the user with the largest instantaneous rate, and max-min fair scheduling
prioritizes the user with the lowest long-term average throughput.
These strategies all operate along the throughput boundary, rendering
them Pareto optimal, a state where no user's throughput can be improved
without degrading others. In contrast, the BC-PFS in O-RAN enables
trustworthy cooperation among multiple operators through blockchain
and achieves higher average user throughput than single-network scheduling.
This gain is further amplified as the number of cooperating networks
increases, which highlights the superiority of the BC-PFS in facilitating
trustworthy inter-network coordination.
\begin{figure}
\centering\includegraphics[width=0.34\paperwidth]{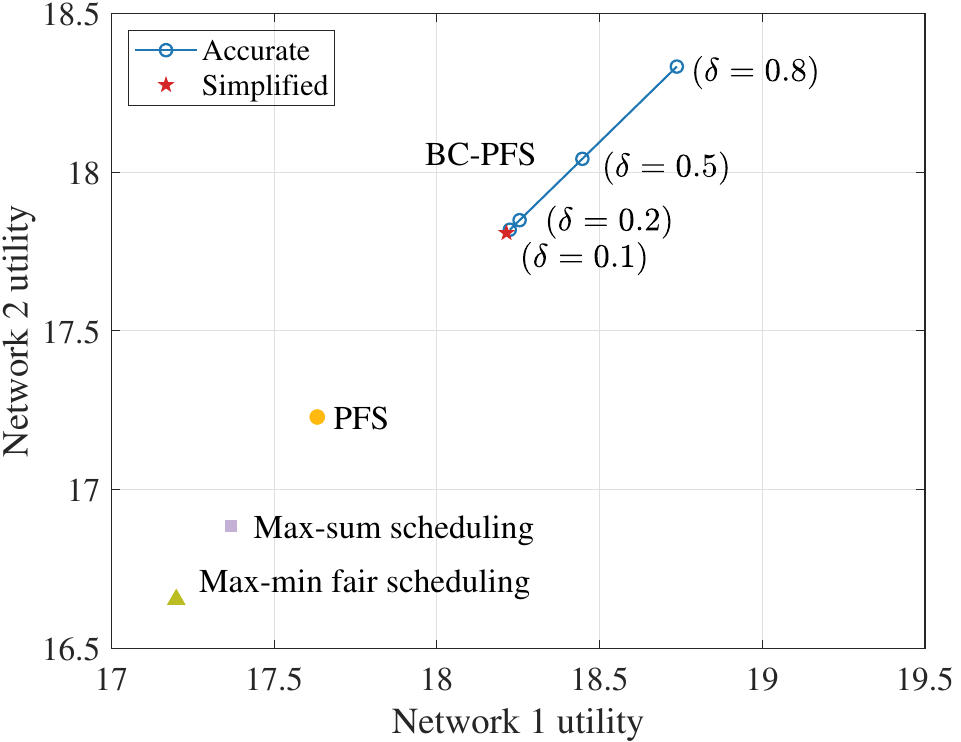}

\caption{\foreignlanguage{american}{Utility region: case of two networks.\label{fig:Utility-region:-case}}}
\vspace{-2mm}
\end{figure}

Concerning the performance advantage of the BC-PFS for O-RAN, Fig.
\ref{fig:Utility-region:-case} illustrates the utility regions achieved
by various scheduling approaches in a two-operator case, where each
operator serves two users. For non-cooperative scheduling strategies,
as theoretically expected, the PFS achieves superior utility performance
compared to both max-sum scheduling and max-min fair scheduling, confirming
its well-known proportional fairness property. More importantly, the
BC-PFS shows significant utility improvement over these conventional
non-cooperative scheduling schemes, demonstrating the advantage of
blockchain-enabled coordination.
\begin{figure}[t]
\centering\subfigure[]{\includegraphics[width=0.33\paperwidth]{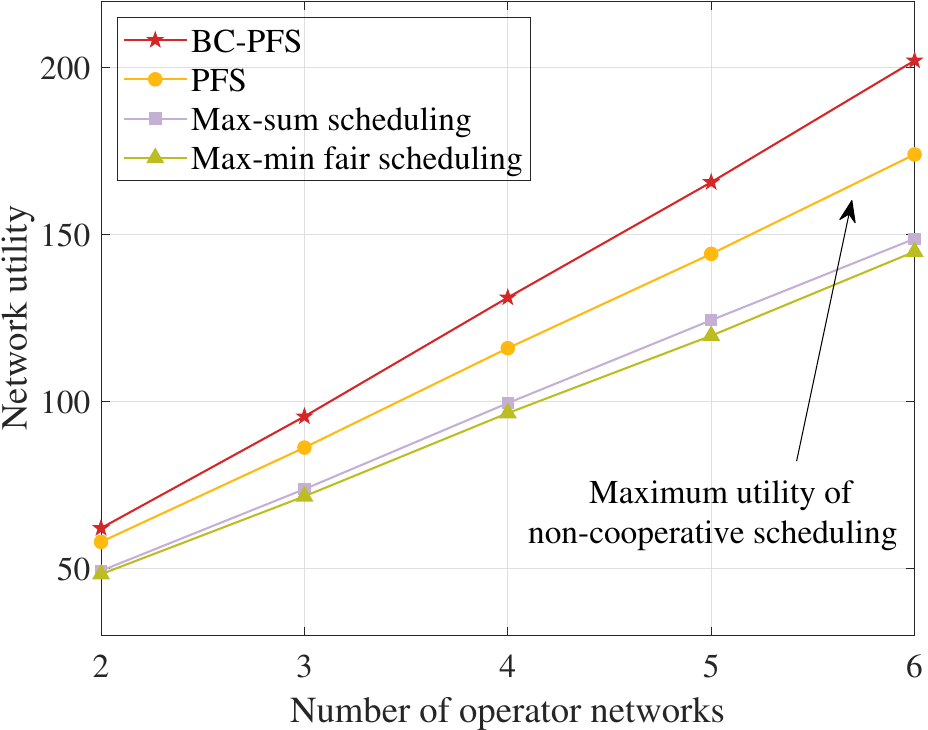}}\hfil

\subfigure[]{\includegraphics[width=0.33\paperwidth]{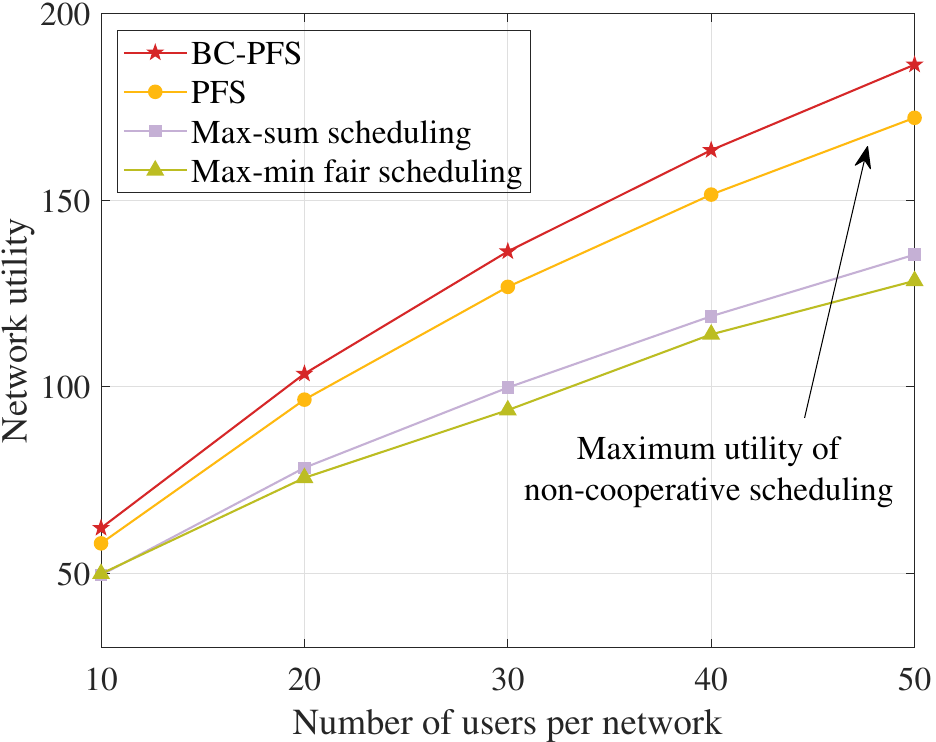}}\caption{\foreignlanguage{american}{Impact of network parameters on network utility. (a) Number of operator
networks. (b) Number of users per network.\label{fig:Comparison-of-network}}}
\vspace{-2mm}
\end{figure}

The simulation further reveals that the network utility of the BC-PFS
in O-RAN grows significantly with increasing rate fluctuation range,
namely $\delta$. This occurs because larger rate variations enhance
multi-operator diversity opportunities, which the BC-PFS effectively
harnesses through its inter-network coordination. This observation
also aligns with our theoretical analysis in Section \ref{sec:Blockchain-Pooling-Effect},
where we established that \eqref{eq:20} represents the lower bound
of network utility, i.e., the red point in Fig. \ref{fig:Utility-region:-case}.
Notably, for non-cooperative scheduling among independent networks,
each operator operates in isolation and their scheduling decisions
cannot exploit diversity from rate variations, resulting in lower
network utility.

Furthermore, Fig. \ref{fig:Comparison-of-network} compares the network
utility achieved by different scheduling approaches under various
network configurations. Fig. \ref{fig:Comparison-of-network}(a) demonstrates
the variation of network utility with the number of operator networks
$K$ under a fixed $N=10$. The results show that in traditional non-cooperative
scheduling frameworks, the PFS consistently maintains optimal utility
performance due to its proportional fairness property, which defines
the performance upper bound for non-cooperative strategies. In contrast,
the BC-PFS in O-RAN demonstrates significant performance advantages
across different $K$ through blockchain-enabled inter-network coordination.
Notably, the performance gap between the BC-PFS and PFS monotonically
increases with the growing number of participating networks, which
validates the scaling effect of multi-network cooperation.

Fig. \ref{fig:Comparison-of-network}(b) focuses on the impact of
the number of users per network $N$ on network utility $U$ with
a fixed $K=2$. Similar to the pattern observed in Fig. \ref{fig:Comparison-of-network}(a),
while the utility of all scheduling schemes improves with increasing
$N$, the BC-PFS consistently maintains a stable performance lead,
and this advantage becomes more pronounced as the user scale expands.
This O-RAN performance enhancement via BC-PFS can be recognized as
the pooling gain.
\begin{figure}
\centering\includegraphics[width=0.33\paperwidth]{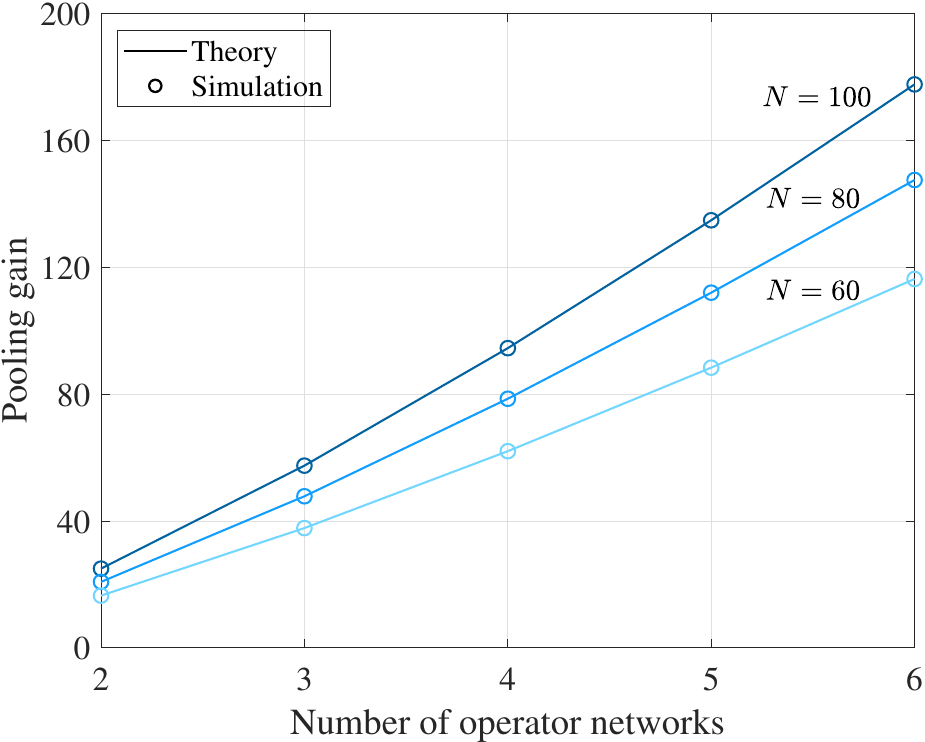}

\caption{\foreignlanguage{american}{Simulation and theoretical results of pooling gain.\label{fig:Simulation-and-theoretical}}}
\vspace{2mm}
\end{figure}

Fig. \ref{fig:Simulation-and-theoretical} presents the pooling gain
under varying numbers of operators $K$ and users per network $N$.
The results demonstrate strong agreement between theoretical and simulated
values. Specifically, when maintaining constant rate fluctuation conditions,
the pooling gain exhibits monotonic improvement with both increasing
$K$ and $N$, which is consistent with Theorem \ref{thm:monotonicity}.
These findings not only validate our theoretical framework, but also
reveal that the pooling effect becomes increasingly pronounced with
growing openness.
\begin{figure}
\centering\includegraphics[width=0.35\paperwidth]{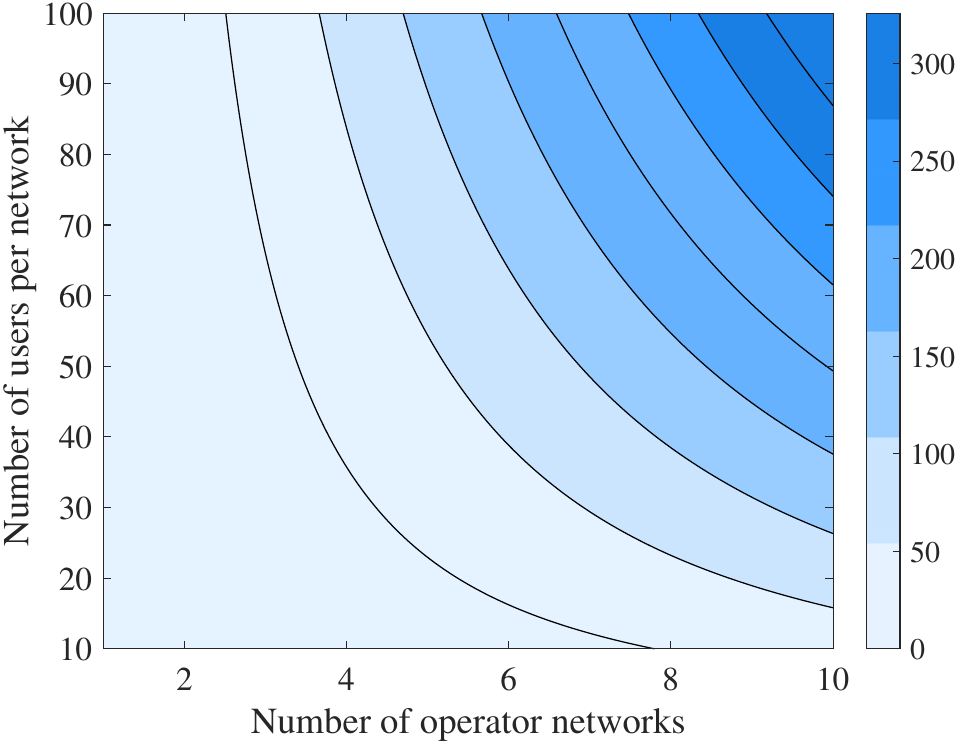}\caption{\foreignlanguage{american}{Pooling gain under different network settings.\label{fig:Blockchain-pooling-effect}}}
\end{figure}

Furthermore, Fig. \ref{fig:Blockchain-pooling-effect} visualizes
the pooling gain\textquoteright s dependence on the number of users
per network and network numbers through a contour plot, where darker
shades indicate higher gains. This gradient-based representation confirms
the monotonic growth property derived in Section VI. The pooling gain
consistently intensifies as network expands. Notably, the contour
spacing becomes progressively uniform at larger network scales. The
initially nonlinear growth of pooling gain gradually converges toward
linearity as network expands, which can be analytically characterized
by \eqref{eq:21}. When $N$ and $K$ are small, the logarithmic term
dominates, creating nonlinear growth; as network scale increases,
the ratio $\frac{\omega\left(KN\right)}{\omega\left(N\right)}$ approaches
a constant via harmonic number approximation, causing $\psi\left(N,K\right)$
to exhibit quasi-linear scaling with respect to $KN$. These insights
show that trustworthy resource pooling in O-RAN achieved through BC-PFS
can yield predictable, near-linear performance gains as the network
scales.

\section{Conclusion\label{sec:Conclusion}}

This work has presented the BC-PFS, a promising solution to establishing
trustworthy cross-operator user scheduling in O-RAN. By integrating
blockchain with the PFS, we have developed a decentralized framework
based on smart contracts to execute the PFS on-chain, enhancing trust
and traceability for O-RAN resource management. Meanwhile, we have
established both accurate and simplified closed-form expressions for
average throughput, offering key insights into system performance.
Our quantitative analysis shows that the pooling gain increases monotonically
with both the number of networks and users. Simulations have validated
the theoretical analysis and confirmed the superiority of the BC-PFS
over non-cooperative approaches for O-RAN. 

\bibliographystyle{IEEEtran}
\bibliography{IEEEabrv,BC-PFS-Part1_arxiv}

\begin{thebibliography}{10}
\providecommand{\url}[1]{#1}
\csname url@samestyle\endcsname
\providecommand{\newblock}{\relax}
\providecommand{\bibinfo}[2]{#2}
\providecommand{\BIBentrySTDinterwordspacing}{\spaceskip=0pt\relax}
\providecommand{\BIBentryALTinterwordstretchfactor}{4}
\providecommand{\BIBentryALTinterwordspacing}{\spaceskip=\fontdimen2\font plus
\BIBentryALTinterwordstretchfactor\fontdimen3\font minus
  \fontdimen4\font\relax}
\providecommand{\BIBforeignlanguage}[2]{{%
\expandafter\ifx\csname l@#1\endcsname\relax
\typeout{** WARNING: IEEEtran.bst: No hyphenation pattern has been}%
\typeout{** loaded for the language `#1'. Using the pattern for}%
\typeout{** the default language instead.}%
\else
\language=\csname l@#1\endcsname
\fi
#2}}
\providecommand{\BIBdecl}{\relax}
\BIBdecl

\bibitem{Fu2024}
Y.~Fu, X.~Wang, and F.~Fang, ``Multi-objective multi-dimensional resource
  allocation for categorized {QoS} provisioning in beyond {5G} and {6G} radio
  access networks,'' \emph{IEEE Trans. Commun.}, vol.~72, no.~3, pp.
  1790--1803, Mar. 2024.

\bibitem{whitePaper2018}
\BIBentryALTinterwordspacing
{O-RAN Alliance}, ``{O-RAN}: Towards an open and smart {RAN},'' White Paper,
  Oct. 2018. [Online]. Available:
  \url{https://mediastorage.o-ran.org/white-papers/O-RAN.White-Paper-2018-10.pdf}
\BIBentrySTDinterwordspacing

\bibitem{Polese2023}
M.~Polese, L.~Bonati, S.~D'Oro, S.~Basagni, and T.~Melodia, ``Understanding
  {O-RAN}: Architecture, interfaces, algorithms, security, and research
  challenges,'' \emph{IEEE Commun. Surv. Tutorials}, vol.~25, no.~2, pp.
  1376--1411, Jun. 2023.

\bibitem{Oransharing2023}
\BIBentryALTinterwordspacing
{O-RAN Alliance}, ``Spectrum sharing based on shared {O-RUs},'' {O-RAN} Next
  Generation Research Group Research Report, Oct. 2023. [Online]. Available:
  \url{https://mediastorage.o-ran.org/ngrg-rr/nGRG-RR-2023-05-Spectrum_Sharing_with_Shared_O-RU-v1_0.pdf}
\BIBentrySTDinterwordspacing

\bibitem{Damnjanovic2024}
A.~Damnjanovic, D.~Knisley, A.~Saurabh, R.~Prakash, X.~Zhang, and S.~Chen,
  ``Spectrum sharing with {O-RAN} architecture,'' in \emph{Proc. IEEE Int.
  Symp. Dyn. Spectr. Access Netw. (DySPAN)}, Washington, DC, USA, May 2024, pp.
  108--113.

\bibitem{Javed2025}
F.~Javed, J.~Mangues-Bafalluy, E.~Zeydan, and L.~Blanco, ``Trustworthy
  reputation for federated learning in {O-RAN} using blockchain and smart
  contracts,'' \emph{IEEE Open J. Commun. Soc.}, vol.~6, pp. 1343--1362, Feb.
  2025.

\bibitem{Giupponi2022}
L.~Giupponi and F.~Wilhelmi, ``Blockchain-enabled network sharing for {O-RAN}
  in {5G} and beyond,'' \emph{IEEE Netw.}, vol.~36, no.~4, pp. 218--225, Aug.
  2022.

\bibitem{Astely2009}
D.~Astely, E.~Dahlman, A.~Furusk\"{a}r, Y.~Jading, M.~Lindstr\"{o}m, and
  S.~Parkvall, ``{LTE}: the evolution of mobile broadband,'' \emph{IEEE Commun.
  Mag.}, vol.~47, no.~4, pp. 44--51, Apr. 2009.

\bibitem{Haque2023}
M.~E. Haque, F.~Tariq, M.~R.~A. Khandaker, K.-K. Wong, and Y.~Zhang, ``A survey
  of scheduling in {5G} {URLLC} and outlook for emerging {6G} systems,''
  \emph{IEEE Access}, vol.~11, pp. 34\,372--34\,396, Apr. 2023.

\bibitem{Zhou2011}
H.~Zhou, P.~Fan, and J.~Li, ``Global proportional fair scheduling for networks
  with multiple base stations,'' \emph{IEEE Trans. Veh. Technol.}, vol.~60,
  no.~4, pp. 1867--1879, Feb. 2011.

\bibitem{Gu2016}
J.~Gu, S.~J. Bae, S.~F. Hasan, and M.~Y. Chung, ``Heuristic algorithm for
  proportional fair scheduling in {D2D}-cellular systems,'' \emph{IEEE Trans.
  Wireless Commun.}, vol.~15, no.~1, pp. 769--780, Jan. 2016.

\bibitem{Li2018}
X.~Li, R.~Shankaran, M.~A. Orgun, G.~Fang, and Y.~Xu, ``Resource allocation for
  underlay {D2D} communication with proportional fairness,'' \emph{IEEE Trans.
  Veh. Technol.}, vol.~67, no.~7, pp. 6244--6258, Jul. 2018.

\bibitem{Zhang2022a}
M.~Zhang, Y.~Guo, L.~Sala\"{u}n, C.~W. Sung, and C.~S. Chen, ``Proportional
  fair scheduling for downlink mmwave multi-user {MISO-NOMA} systems,''
  \emph{IEEE Trans. Veh. Technol.}, vol.~71, no.~6, pp. 6308--6321, Jun. 2022.

\bibitem{Faisal2022}
T.~Faisal, M.~Dohler, S.~Mangiante, and D.~R. Lopez, ``{BEAT}:
  Blockchain-enabled accountable and transparent network sharing in {6G},''
  \emph{IEEE Commun. Mag.}, vol.~60, no.~4, pp. 52--56, 2022.

\bibitem{Wu2026}
M.~Wu, X.~Ling, J.~Wang, Y.~Le, K.~Huang, B.~Cao, Y.~Huang, D.~Niyato, Z.~Ding,
  and X.~You, ``Blockchain-driven resource management in wireless
  communications and networks: Models, approaches, and applications,''
  \emph{IEEE Commun. Surv. Tutorials}, vol.~28, pp. 2306--2344, Aug. 2026.

\bibitem{Xu2024}
H.~Xu, Z.~Zhou, L.~Zhang, Y.~Sun, and C.-L. I, ``{BE-RAN}: Blockchain-enabled
  {Open RAN} for {6G} with {DID} and privacy-preserving communication,'' in
  \emph{Proc. IEEE Global Commun. Conf. Workshops (GLOBECOM Workshops)}, Cape
  Town, SA, Dec. 2024, pp. 1015--1021.

\bibitem{Ling2019}
X.~Ling, J.~Wang, T.~Bouchoucha, B.~C. Levy, and Z.~Ding, ``Blockchain radio
  access network ({B-RAN}): Towards decentralized secure radio access
  paradigm,'' \emph{IEEE Access}, vol.~7, pp. 9714--9723, Jan. 2019.

\bibitem{Ling2025}
X.~Ling, Y.~Le, J.~Wang, Y.~Huang, and X.~You, ``Trust and trustworthiness in
  information and communications technologies,'' \emph{IEEE Wireless Commun.},
  vol.~32, no.~2, pp. 84--92, Apr. 2025.

\bibitem{Chen2025}
B.~Chen, X.~Ling, W.~Cao, J.~Wang, and Z.~Ding, ``Analysis of channel
  uncertainty in trusted wireless services via repeated interactions,''
  \emph{IEEE J. Sel. Areas Commun.}, vol.~43, no.~6, pp. 2248--2265, Jun. 2025.

\bibitem{Wang2022d}
Z.~Wang, B.~Cao, C.~Liu, C.~Xu, and L.~Zhang, ``Blockchain-based fog radio
  access networks: Architecture, key technologies, and challenges,''
  \emph{Digital Commun. Networks}, vol.~8, no.~5, pp. 720--726, Oct. 2022.

\bibitem{Ling2020a}
X.~{Ling}, J.~{Wang}, Y.~{Le}, Z.~{Ding}, and X.~{Gao}, ``Blockchain radio
  access network beyond 5{G},'' \emph{IEEE Wireless Commun.}, vol.~27, no.~6,
  pp. 160--168, Dec. 2020.

\bibitem{Cao2023}
W.~Cao, X.~Ling, J.~Wang, Z.~Ding, and X.~Gao, ``A framework for
  {QoS}-guaranteed fast access services in blockchain radio access network,''
  \emph{IEEE Trans. Wireless Commun.}, vol.~23, no.~4, pp. 2711 -- 2725, Apr.
  2024.

\bibitem{Ling2025a}
X.~Ling, Y.~Le, S.~Chen, J.~Wang, and X.~Zhou, ``Blockchain-enabled
  decentralized services and networks: Assessing roles and impacts,''
  \emph{IEEE J. Sel. Areas Commun.}, vol.~43, no.~6, pp. 2141--2154, Jun. 2025.

\bibitem{Femenias2024}
G.~Femenias, M.~Francisca~Hinarejos, F.~Riera-Palou, J.-L. Ferrer-Gomila, and
  A.~Jaume-Barcel\'{o}, ``Dynamic spectrum sharing in a blockchain enabled
  network with multiple cell-free massive {MIMO} virtual operators,''
  \emph{IEEE Access}, vol.~12, pp. 70\,615--70\,633, May 2024.

\bibitem{Al2024}
A.~Al-Khatib, H.~Hadi, H.~Timinger, and K.~Moessner, ``Blockchain-empowered
  resource trading for optimizing bandwidth reservation in vehicular
  networks,'' \emph{IEEE Access}, vol.~12, pp. 90\,084--90\,098, Jun. 2024.

\bibitem{Qian2025}
L.~Qian, C.~Liu, and J.~Zhao, ``User connection and resource allocation
  optimization in blockchain empowered metaverse over {6G} wireless
  communications,'' \emph{IEEE Trans. Wireless Commun.}, vol.~24, no.~1, pp.
  19--34, 2025.

\bibitem{Zhang2020d}
H.~Zhang, S.~Leng, and H.~Chai, ``A blockchain enhanced dynamic spectrum
  sharing model based on proof-of-strategy,'' in \emph{Proc. IEEE Int. Conf.
  Commun. (ICC)}, Dublin, IE, Jul. 2020, pp. 1--6.

\bibitem{Yu2026}
Z.~Yu, Z.~Chang, T.~Mikkonen, V.~Frascolla, and S.~Mumtaz, ``Blockchain
  cooperative spectrum management for {Wi-Fi} and {LTE}-unlicensed coexistence
  networks,'' \emph{IEEE Trans. Commun.}, vol.~74, pp. 6409--6425, Mar. 2026.

\bibitem{Zhang2022}
H.~Zhang, S.~Leng, Y.~Wei, and J.~He, ``A blockchain enhanced coexistence of
  heterogeneous networks on unlicensed spectrum,'' \emph{IEEE Trans. Veh.
  Technol.}, vol.~71, no.~7, pp. 7613--7624, Apr. 2022.

\bibitem{Zhu2022}
R.~Zhu, H.~Liu, L.~Liu, X.~Liu, W.~Hu, and B.~Yuan, ``A blockchain-based
  two-stage secure spectrum intelligent sensing and sharing auction
  mechanism,'' \emph{IEEE Trans. Ind. Inf.}, vol.~18, no.~4, pp. 2773--2783,
  Aug. 2022.

\bibitem{Ayepah2023}
D.~Ayepah-Mensah, G.~Sun, G.~O. Boateng, S.~Anokye, and G.~Liu,
  ``Blockchain-enabled federated learning-based resource allocation and trading
  for network slicing in {5G},'' \emph{IEEE/ACM Trans. Networking}, vol.~32,
  no.~1, pp. 654--669, Feb. 2024.

\bibitem{Rappaport2001}
T.~Rappaport, \emph{Wireless Communications: Principles and Practice},
  2nd~ed.\hskip 1em plus 0.5em minus 0.4em\relax USA: Prentice Hall PTR, Dec.
  2001.

\bibitem{Kushner2004}
H.~Kushner and P.~Whiting, ``Convergence of proportional-fair sharing
  algorithms under general conditions,'' \emph{IEEE Trans. Wireless Commun.},
  vol.~3, no.~4, pp. 1250--1259, Jul. 2004.

\bibitem{Liu2011}
E.~Liu, Q.~Zhang, and K.~K. Leung, ``Asymptotic analysis of proportionally fair
  scheduling in {Rayleigh} fading,'' \emph{IEEE Trans. Wireless Commun.},
  vol.~10, no.~6, pp. 1764--1775, Jun. 2011.

\bibitem{Liew2008}
S.~C. Liew and Y.~J. Zhang, ``Proportional fairness in multi-channel multi-rate
  wireless networks-part {I}: The case of deterministic channels with
  application to {AP} association problem in large-scale {WLAN},'' \emph{IEEE
  Trans. Wireless Commun.}, vol.~7, no.~9, pp. 3446--3456, Sep. 2008.

\bibitem{Ma2019}
H.~Ma, J.~Cheng, and X.~Wang, ``Proportional fair secrecy beamforming for
  {MISO} heterogeneous cellular networks with wireless information and power
  transfer,'' \emph{IEEE Trans. Commun.}, vol.~67, no.~8, pp. 5659--5673, Aug.
  2019.

\bibitem{kelly1997}
F.~Kelly, ``Charging and rate control for elastic traffic,'' \emph{Eur. Trans.
  Telecommun.}, vol.~8, no.~1, pp. 33--37, Jan. 1997.

\bibitem{Luo2024}
H.~Luo, X.~Yang, H.~Yu, G.~Sun, B.~Lei, and M.~Guizani, ``Performance analysis
  and comparison of nonideal wireless {PBFT} and {RAFT} consensus networks in
  {6G} communications,'' \emph{IEEE Internet Things J.}, vol.~11, no.~6, pp.
  9752--9765, Mar. 2024.

\bibitem{kaspa2025}
\BIBentryALTinterwordspacing
Kaspa. [Online]. Available: \url{https://kaspa.org/}
\BIBentrySTDinterwordspacing

\bibitem{Quorum2018}
\BIBentryALTinterwordspacing
A.~Baliga, I.~Subhod, P.~Kamat, and S.~Chatterjee, ``Performance evaluation of
  the quorum blockchain platform,'' \emph{arXiv:1809.03421}, Oct. 2018.
  [Online]. Available: \url{https://arxiv.org/abs/1809.03421}
\BIBentrySTDinterwordspacing

\bibitem{megaeth2025}
\BIBentryALTinterwordspacing
MegaEth. [Online]. Available: \url{https://www.megaeth.com/}
\BIBentrySTDinterwordspacing

\bibitem{Jung2011}
H.~Jung, T.~â. Kwon, K.~Cho, and Y.~Choi, ``{REACT}: Rate adaptation using
  coherence time in 802.11 {WLANs},'' \emph{Comput. Commun.}, vol.~34, no.~11,
  pp. 1316--1327, Jul. 2011.

\end{thebibliography}

\end{document}